\documentclass[submission,copyright,creativecommons]{eptcs}
\providecommand{\event}{LFMTP 2025} % Name of the event you are submitting to

\usepackage{iftex}

\ifpdf
  \usepackage{underscore}         % Only needed if you use pdflatex.
  \usepackage[T1]{fontenc}        % Recommended with pdflatex
\else
  \usepackage{breakurl}           % Not needed if you use pdflatex only.
\fi

\title{Type Theory with Single Substitutions\thanks{This research was supported
by the European Union (ERC, HOTT, 101170308). Views and opinions expressed are
however those of the author(s) only and do not necessarily reflect those of the
European Union or the European Research Council. Neither the European Union nor
the granting authority can be held responsible for them.}}
\author{Ambrus Kaposi
\institute{Eötvös Loránd University (ELTE)\\
Budapest, Hungary}
\email{akaposi@inf.elte.hu}
\and
Szumi Xie
\institute{Eötvös Loránd University (ELTE)\\
Budapest, Hungary}
\email{szumi@inf.elte.hu}
}
\def\titlerunning{Type Theory with Single Substitutions}
\def\authorrunning{A. Kaposi \& Sz. Xie}

\usepackage{tikz}
\usetikzlibrary{arrows}
\usepackage{proof}
\usepackage{multicol}
\usepackage{amsmath}
\usepackage{amsthm}
\usepackage{amssymb}

\newtheorem{theorem}{Theorem}
\newtheorem{problem}[theorem]{Problem}
\newtheorem{definition}[theorem]{Definition}
\newtheorem{lemma}[theorem]{Lemma}

\newcommand{\Agda}{%
  \raisebox{-.05em}{\includegraphics{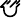}}%
}

\newcommand{\ra}{\rightarrow}
\newcommand{\Ra}{\Rightarrow}
\newcommand{\Set}{\mathsf{Set}}
\newcommand{\Prop}{\mathsf{Prop}}
\newcommand{\Ty}{\mathsf{Ty}}
\newcommand{\Tm}{\mathsf{Tm}}
\newcommand{\Con}{\mathsf{Con}}
\newcommand{\Sub}{\mathsf{Sub}}
\newcommand{\Tel}{\mathsf{Tel}}
\newcommand{\Tms}{\mathsf{Tms}}
\newcommand{\p}{\mathsf{p}}
\newcommand{\q}{\mathsf{q}}
\newcommand{\ext}{\mathop{\triangleright}}
\newcommand{\N}{\mathbb{N}}
\newcommand{\lam}{\mathsf{lam}}
\newcommand{\app}{\mathsf{app}}
\newcommand{\U}{\mathsf{U}}
\newcommand{\El}{\mathsf{El}}
\newcommand{\cd}{\mathsf{c}}
\newcommand{\blank}{\mathord{\hspace{1pt}\text{--}\hspace{1pt}}} %from the book
\renewcommand{\tt}{\mathsf{tt}}
\newcommand{\fst}{\mathsf{fst}}
\newcommand{\snd}{\mathsf{snd}}
\newcommand{\Lift}{\mathsf{Lift}}
\newcommand{\mk}{\mathsf{mk}}
\newcommand{\un}{\mathsf{un}}
\newcommand{\id}{\mathsf{id}}

\begin{document}
\maketitle

\begin{abstract}
Type theory can be described as a generalised algebraic theory. This
automatically gives a notion of model and the existence of the syntax
as the initial model, which is a quotient inductive-inductive
type. Algebraic definitions of type theory include Ehrhard's
definition of model, categories with families (CwFs), contextual
categories, Awodey's natural models, C-systems, B-systems. With the
exception of B-systems, these notions are based on a parallel
substitution calculus where substitutions form a category. In this
paper we define a single substitution calculus (SSC) for type theory
and show that the SSC syntax and the CwF syntax are isomorphic for a
theory with dependent function space and a hierarchy of universes. SSC
only includes single substitutions and single weakenings, and eight
equations relating these: four equations describe how to substitute
variables, and there are four equations on types which are needed to
typecheck the other equations. SSC provides a simple, minimalistic
alternative to parallel substitution calculi or B-systems for defining
type theory. SSC relates to CwF as extensional combinatory calculus
relates to lambda calculus: there are more models of the former, but
the syntaxes are equivalent. If we have some additional type formers,
we show that an SSC model gives rise to a CwF.
\end{abstract}

\section{Introduction}

What is type theory? Here we refer to type theory as a particular
formal system based on Martin-Löf's original definition
\cite{MARTINLOF197573}, and not to the study of type systems
(e.g.\ \cite{DBLP:books/daglib/0005958}).

In this paper, we will answer this question by defining type theory as
a generalised algebraic theory (GAT
\cite{DBLP:journals/apal/Cartmell86}). GATs are multi-sorted algebraic
theories where later sorts can be indexed by previous sorts. An
example is categories where the sort of morphisms is double-indexed
over objects.
The GAT presentation of type theory is intrinsic
\cite{DBLP:conf/csl/AltenkirchR99,DBLP:conf/popl/AltenkirchK16} rather
than extrinsic
\cite{alti:phd93,abel2013normalization,DBLP:conf/cpp/AdjedjLMPP24}:
this means that we only consider well-formed, well-scoped and
well-typed abstract syntax trees, there are no meaningless terms; in
the GAT presentation, the conversion relation is expressed by
equations of the algebraic theory.

The GAT presentation has the following advantages:
\begin{itemize}
\item In type theory, typing depends on conversion: in an extrinsic
  presentation, this is expressed by the rule deriving $\Gamma\vdash t
  : B$ from $\Gamma\vdash A \sim B$ and $\Gamma\vdash t : A$.  When
  conversion $\sim$ is equality, this rule holds by the general
  properties of equality; and we also avoid lots of boilerplate
  stating that e.g.\ conversion is a congruence with respect to all
  operations (for an alternative intrinsic presentation with explicit conversion
  relation see \cite{chapman09eatitself}).
\item Syntax and semantics are concordant: the syntax is simply the
  initial model, which always exists for any GAT and is called a
  quotient inductive-inductive type
  \cite{DBLP:journals/pacmpl/KaposiKA19}. The fact that the syntax is
  quotiented means that any function defined on it needs to respect
  the equations. For example, it is not possible to define a printing
  function which distinguishes the terms $(\lambda x\ldotp x)\,(\lambda y\ldotp y)$ and $\lambda y\ldotp y$,
  as these are convertible. However, it is still possible to define
  normalisation \cite{lmcs:4005} and typechecking \cite{typecheck} at
  this level of abstraction.
\item The GAT approach is also more abstract than extrinsic
  approaches, which means that we don't have to make ad-hoc
  choices. For example, we do not have to decide whether we want
  function space {\`a} la Curry or {\`a} la Church; whether we work
  with paranoid or economic typing rules \cite[Section
    5]{winterhalter:hal-01626651}; whether we have explicit or
  recursively defined substitution (instantiation of variables by
  terms, written $t[x\mapsto u]$ or $t[u/x]$).  These choices are
  forced by the notion of GAT: the domain of a lambda appears in its
  type, so we are always {\`a} la Church; algebraic theories don't
  allow recursively defined operators, so instantiation is an explicit
  operation; and so on.
\end{itemize}

% For formalisation on a computer, the algebraic approach also has its
% downsides: as mutually defined recursive operations are not allowed,
% substitution (instantiation of variables by terms) is explicit, which
% means that it is a constructor, and does not compute by
% definition. For example, $(t\,a)[x\mapsto u] = (t[x\mapsto
%   u])\,(a[x\mapsto u])$ is a derivable equation in the syntax, but it
% is not definitional. This makes formalisation difficult, and results
% in the situation called ``transport hell''.

However, even at the GAT level of abstraction, there are some choices
to be made about the instantiation operation. It is usually convenient
to use parallel substitutions, which means that the instantiation
operation replaces all variables in a context by terms in another
context \emph{at the same time}, that is, $t[x_0\mapsto
  u_0,\dots,x_n\mapsto u_n]$ where $x_0,\dots,x_n$ are the variables
in the context of $t$. We call $(x_0\mapsto u_0,\dots,x_n\mapsto u_n)$
a parallel substitution, where the terms $u_0,\dots,u_n$ are
all in the same context. Parallel substitutions are very natural in
the algebraic setting because they form the morphisms in a category
where objects are contexts (lists of types), then types/terms are
presheaves, context extension is a representability condition, so the
equations of the algebraic theory are forced by the categorical
structure. Some examples of algebraic definitions of type theory using
parallel substitutions ranging from the more syntactic to the more
semantic: Ehrhard's notion \cite{ehrhard,coquandEhrhard}, categories
with families (CwF \cite{DBLP:conf/types/Dybjer95,Castellan2021}),
natural models \cite{DBLP:journals/mscs/Awodey18}, contextual
categories \cite{DBLP:journals/apal/Cartmell86}, C-systems
\cite{DBLP:journals/lmcs/AhrensLV18}, locally cartesian closed
categories \cite{DBLP:journals/mscs/ClairambaultD14}, path categories
\cite{DBLP:journals/tocl/Berg18}.

In traditional extrinsic presentations
(e.g.\ \cite{alti:phd93,DBLP:books/daglib/0005958}) instantiation
replaces free occurrences of a single variable $x$ with a term $u$ as
in the notation $t[x\mapsto u]$. A substitution consists of a
variable name (a pointer into a context) and a term. Substitutions
cannot be composed, thus they do not form a category. To bridge the
gap between the intrinsic parallel substitution world and the
extrinsic single substitution world, Voevodsky introduced B-systems
\cite{bc} which are an algebraic
description of a single substitution calculus. B-systems involve
complex rules describing telescopes, weakenings and substitutions
under telescopes, with several equations.

In this paper, we introduce a new single substitution calculus (SSC)
which is simpler and more minimalistic than existing algebraic
approaches. The instantiation operation $\blank[\blank]$ takes either a
single substitution or a single weakening as an argument, and there
are eight equations explaining how these behave: four of these say how
to instantiate variables and the other four equations are needed to
typecheck the first four. The SSC shows that to explain type theory in
an algebraic way, there is no need for categories, parallel
substitutions or weakenings, empty substitution, telescopes, or
combinations of these. All our operations are easy to motivate, we
showcase this in Section \ref{sec:tt} which is a tutorial introduction
to our theory. We believe that our SSC is in a sweet spot: our
notations are close to traditional extrinsic notations, but they are
algebraic, thus come with a well behaved model theory, and can be
easily related to more semantic descriptions.

Our SSC is in some sense too minimalistic: for example, the equation
$b[\p^+][\langle\q\rangle] = b$ is not \emph{derivable} in any model,
but it is \emph{admissible} in the syntax. (Here $b : \Tm\,(\Gamma\ext
A)\,B$ is a term which depends on a context $\Gamma$ and an extra
variable of type $A$; $b[\p^+] : \Tm\,(\Gamma\ext A\ext
A[\p])\,(B[\p^+])$ is a version where we weakened \emph{before} the
last variable; in $b[\p^+][\langle\q\rangle] : \Tm\,(\Gamma\ext A)\,B$
we substituted the last variable for the previous one.)  This is
analogous to e.g.\ parametricity results
\cite{DBLP:journals/jfp/BernardyJP12}, which do not hold in an
arbitrary model, but they hold in the syntax. In this paper, we show
that the \emph{syntaxes} of SSC and the corresponding parallel
substitution calculus (CwF-based calculus) are isomorphic. All
CwF-based models are models of SSC, but not
necessarily the other way. While B-systems are equivalent to contextual CwFs (and
C-systems \cite{bc}), there are more
SSC models than CwF-based models. This relationship is like that
of extensional combinatory calculus and lambda calculus, where the
syntaxes are equivalent, and every lambda model is a combinatory
model, but not the other way
\cite{DBLP:conf/fscd/AltenkirchKSV23}. Another example is the
relationship between monoids over a set $X$ and nil/cons algebras over
$X$: every monoid is a nil/cons-algebra, but not the other way;
it is well-known that free monoids over $X$ (syntax of monoid over
$X$) are the same as $X$-lists (the syntax of nil/cons algebras
over $X$).

We introduce a technology for proving admissible rules. We define $\alpha$-normal forms and show that every term admits an $\alpha$-normal form.
Thus, proving something for $\alpha$-normal forms implies it for all terms.
Just as $\beta$-normal forms
don't distinguish between $\beta$-convertible terms, $\alpha$-normal
forms don't include explicit instantiations. $\alpha$-normal forms are
not like ordinary $\beta/\eta$-normal forms because they are 
quotiented by the computation/uniqueness rules like $\beta$/$\eta$ for
function space. We use induction on $\alpha$-normal forms to define
parallel substitutions and prove all CwF rules in the syntax
of SSC.

Our SSC has 8 equations, and these are enough to describe type
theories with arbitrary choice of type formers (see the paragraph on
SOGATs in Subsection \ref{sec:related}). In the particular case of
type theory with a Coquand-universe and $\Pi$ types, we can do better:
4 (conditional) equations are enough, because we can derive the
other 4 equations using the fact that every type is represented by a
term.

If in addition to having a universe and $\Pi$ types, our theory
also supports unit and $\Sigma$ types, we can actually derive a
CwF-based model from an SSC-model. The idea is that $\Sigma$ types can
represent contexts, functions between $\Sigma$ types represent
parallel substitutions, functions into the universe represent
dependent types. Then the SSC is just there for bootstrapping
purposes: SSC formulates $\Sigma$, $\Pi$, $\U$, and then we use these
to define a parallel substitution calculus. In summary, although SSC
models without any type formers are weaker than CwFs, if we equip them
with enough type formers, the SSC model gives rise to a CwF. Doing a
roundtrip CwF $\longrightarrow$ SSC $\longrightarrow$ CwF results in a
contextually isomorphic CwF.

All of the constructions in this paper can be understood as happening
internal to extensional type theory, and most results were formalised
in Agda (see Subsection \ref{sec:metatheory}). Some results in this
paper were presented at the TYPES 2024 conference \cite{singleTypes}.

In summary, our contributions are (in the same order as the structure of the paper):
\begin{description}
\item[Section \ref{sec:tt}] A new generalised algebraic presentation of type theory in the
  form of a minimalistic single substitution calculus. Our calculus
  does not feature parallel substitutions/weakenings, empty
  substitution, telescopes. We present our calculus in an
  easy-to-understand way where all operations are well-motivated.
\item[Section \ref{sec:admissible}] The $\alpha$-normalisation technique which shows that the syntax
  of our single substitution calculus is isomorphic to the syntax of
  the CwF-based theory.
\item[Section \ref{sec:minimisation}] For type theory with $\Pi$, $\U$, a minimised presentation of
  the equations which results in an isomorphic theory.
\item[Section \ref{sec:cwf}] For type theory with $\Pi$, $\U$, $\top$,
  $\Sigma$, we show that a CwF structure is derivable.
\end{description}

% \subsection{Structure of the paper}
% 
% After discussing related work and our metatheory, in Section
% \ref{sec:tt}, we introduce our single substitution calculus (SSC) only
% relying on a working knowledge of an implementation of type theory
% such as Agda, no prior experience with the metatheory of type theory
% is required. Section \ref{sec:admissible} shows that in the syntax of
% SSC, several new equations are admissible, and that the syntax of SSC
% is actually isomorphic to the CwF-based syntax. In Section
% \ref{sec:minimisation}, we minimise our SSC obtaining a theory with
% fewer equations, relying on the presence of $\Pi$ and $\U$ in our
% theory. In Section \ref{sec:cwf}, we show that in the presence of
% $\Pi$, $\U$, $\Sigma$, SSC-models are actually equivalent to
% CwF-models in a weak sense. We conclude in Section
% \ref{sec:conclusion}.

\subsection{Related work}\label{sec:related}

\subparagraph*{B-systems.} B-systems introduced by Voevodsky \cite{voevodsky2014bsystems}
are an intrinsic, essentially
algebraic presentation of type theory using single
substitutions. B-systems relate to our SSC as essentially
algebraic theories relate to generalised algebraic theories, or
sets with a map into $I$ relate to indexed families over $I$
\cite[page 221]{DBLP:journals/apal/Cartmell86}. Another difference is
that we have fewer and less general equations, resulting in the fact
that we have more models than B-systems. A further difference is that contexts
in B-systems are inductively defined (we call such models \emph{contextual}).
However in the syntax of our
theory, all the rules of B-systems are admissible. We describe the
relationship in more detail in Appendix \ref{app:bsys}.
Kaposi and Luksa \cite{luksa} defined a telescopic SSC for simple type theory, it can be seen as a simply typed
version of B-systems. They show that the category of contextual models
of their calculus is equivalent to the category of contextual simply
typed CwFs.
%Their equivalence result holds even without the presence of function types.
%In the presence of function space,
If we have function space, the same equivalence holds for the simply typed version of our
substitution calculus (\href{https://szumixie.github.io/single-subst/STT.Contextual.html}{\Agda}).

\vspace{-1em}
\subparagraph*{Ehrhard's calculus.}
The first generalised algebraic presentation of type theory was
Ehrhard's calculus \cite{ehrhard,coquandEhrhard} which featured a
parallel substitution calculus with $\p$, $\blank^+$ and
$\langle\blank\rangle$ operations, just like our calculus. Our
single substitution calculus is essentially Ehrhard's calculus with
the categorical composition and identity operations
removed. Categories with families (CwFs
\cite{DBLP:conf/types/Dybjer95,Castellan2021}) feature $\p$, $\q$,
$(\blank,\blank)$ operations which make it more apparent that
substitutions are lists of terms. CwFs are equivalent to Ehrhard's
calculus and the natural models of Awodey \cite{DBLP:journals/mscs/Awodey18}. Contextual versions of these are equivalent to contextual categories/C-systems
\cite{DBLP:journals/apal/Cartmell86,DBLP:journals/lmcs/AhrensLV18}, and
B-systems \cite{bc}. For the precise statements of equivalences, see e.g.\ \cite{DBLP:conf/aplas/AhrensLN24}.

\vspace{-1em}
\subparagraph*{Formalisations of type theory.}
Most computer formalisations of type theory are extrinsic
\cite{DBLP:journals/pacmpl/0001OV18,DBLP:conf/cpp/AdjedjLMPP24,DBLP:journals/jar/SozeauABCFKMTW20}:
they define the syntax as abstract syntax trees, and equip it with
typing and conversion relations. A higher level representation is
working in setoid hell \cite{chapman09eatitself,setoidhell}: terms are
intrinsically well-typed, but conversion is still an explicit
relation. For formalising the GAT-level syntax, one needs a stronger
metatheory than plain Agda or Coq: the metatheory has to support
quotient inductive-inductive types (QIITs
\cite{DBLP:journals/pacmpl/KaposiKA19}), in other words, initial
models of GATs. Altenkirch and Kaposi
\cite{DBLP:conf/popl/AltenkirchK16} formalised the syntax of type
theory using postulated QIITs in Agda, together with parametricity and
normalisation \cite{lmcs:4005} proofs. Brunerie and de Boer
\cite{initiality-agda} constructed the initial contextual category in
Agda using postulated quotients. Although we don't have a general
proof, we have experimental evidence that QIITs are supported by the
cubical set model \cite{DBLP:conf/lics/CoquandHM18} and the setoid
model \cite{kaposi-qiit-setoid}, thus Cubical Agda
\cite{DBLP:journals/jfp/VezzosiMA21} and setoid/observational type
theory \cite{DBLP:conf/mpc/AltenkirchBKT19,DBLP:phd/hal/Pujet22} support QIITs. For example,
in Cubical Agda, set-truncated and groupoid-truncated syntaxes of type
theory have been formalised \cite{cohtt}. Without QIITs, tricks such
as shallowly embedding the syntax can be used to formalise metatheoretic
results about type theory such as gluing
\cite{glue}, and its special cases
canonicity and parametricity \cite{DBLP:conf/mpc/KaposiKK19}.

\vspace{-1em}
\subparagraph*{SOGATs}
We can define languages with binders as second-order theories
following higher-order abstract syntax \cite{DBLP:conf/lics/Hofmann99}
and logical frameworks
\cite{DBLP:journals/jacm/HarperHP93,beluga}. Second-order
generalised algebraic theories (SOGATs
\cite{uemura,DBLP:conf/fscd/KaposiX24}) can be used to present type
theory more abstractly than GATs: SOGATs abstract over the exact
definition of the substitution calculus (whether substitutions are
single or parallel; whether we have CwF-style or Ehrhard-style
operations): contexts, substitutions, instantiations are simply
modelled by the metatheoretic function space. For example, a type
theory with $\Pi$, Coquand-universes \cite{coquandUniverse} and a universe level lifting
operation is described by the following SOGAT (where $\cong$ means
isomorphism):
\begin{equation}\label{eq:tt}
\begin{alignedat}{10}
  & \Ty && : \N\ra\Set                                                             && \U && : (i:\N)\ra\Ty\,(1+i) \\           
  & \Tm && : \Ty\,i\ra\Set                                                         && \El && : \Tm\,(\U\,i)\cong\Ty\,i : \cd \\
  & \Pi && : (A:\Ty\,i)\ra(\Tm\,A\ra\Ty\,i)\ra\Ty\,i                               && \Lift && : \Ty\,i\ra\Ty\,(1+i) \\        
  & \app && : \Tm\,(\Pi\,A\,B)\cong((a:\Tm\,A)\ra\,\Tm\,(B\,a)): \lam \hspace{3em} && \un && : \Tm\,(\Lift\,A)\cong\Tm\,A : \mk
\end{alignedat}
\end{equation}
% Adding $\top$ and $\Sigma$ types (Definition \ref{def:sigma}) is
% described as the following extension of the above SOGAT:
% \begin{equation}\label{eq:sigma}
%   \begin{alignedat}{10}
%     & \top && : \Ty\,0 \\
%     & \tt && : \Tm\,\top \\
%     & \top\eta && : t = \tt \\
%     & \Sigma && : (A:\Ty\,i)\ra(\Tm\,A\ra\Ty\,i)\ra\Ty\,i \\
%     & \fst,\snd && : \Tm\,(\Sigma\,A\,B)\cong(a:\Tm\,A)\times\,\Tm\,(B\,a): \blank,\blank
% \end{alignedat}
% \end{equation}
Second-order models of SOGATs are not well-behaved, as there is no
good notion of homomorphism between them \cite[bottom of page 5]{DBLP:conf/fscd/KaposiX24}. This is why Kaposi and Xie
\cite{DBLP:conf/fscd/KaposiX24} translate SOGATs to GATs: a model of a
SOGAT then is a model of the GAT that we obtain by their
translation. GATs come with well-behaved metatheory: a category of
models with initial model (syntax), (co)free models
\cite{andras,DBLP:phd/hal/Moeneclaey22}, and so on. However, the
translation is not unique: Kaposi and Xie
\cite{DBLP:conf/fscd/KaposiX24} define two different translations, one
based on parallel substitutions and one based on single
substitutions. If we apply their parallel translation to the SOGAT
\vspace{-0.5em}
\begin{equation}\label{eq:tytm}
\begin{alignedat}{10}
  & \Ty && : \Set \hspace{3em} \Tm && : \Ty \ra\Set
\end{alignedat}
\end{equation}
with only two sorts and no operations or equations, we obtain the GAT
of categories with families (CwFs). If we apply their single substitution
translation to the SOGAT (\ref{eq:tt}), we obtain Definition
\ref{def:tt} (also listed in Appendix \ref{app:tt}). The current paper
can be seen as taking a big enough type theory (described as a SOGAT),
and investigating how the result of its single substitution
translation relates to the result of its parallel translation (the
latter first-order theory is quite well-known).
%There is no general result about the relationship of the two translations.

There are techniques to prove properties of type theories without
choosing between parallel or single substitutions, and staying at the
SOGAT level of abstraction: Synthethic Tait Computability
\cite{DBLP:phd/us/Sterling22} and internal sconing
\cite{DBLP:conf/fscd/BocquetKS23}.

\vspace{-0.5em}
\subsection{Metatheory and formalisation}\label{sec:metatheory}

Our metatheory is observational type theory
\cite{DBLP:phd/hal/Pujet22} with quotient inductive-inductive types
(QIITs). On paper we usually omit writing coercions, so our notation
is close to extensional type theory. Definition \ref{def:tt} and Section
\ref{sec:admissible} of this paper are formalised in the proof
assistant Agda. The formalisation is available online
(\href{https://szumixie.github.io/single-subst/TT.index.html}{\Agda}). The Agda
logos next to definitions/theorems are links to their formalised counterparts. We use a strict Prop-valued
\cite{DBLP:journals/pacmpl/GilbertCST19} equality type with postulated
coercion rule (transport rule) and postulated QIITs with computation
rules added using rewrite rules
\cite{DBLP:journals/pacmpl/CockxTW21}. Thus, we work in a setting of
uniqueness of identity proofs (UIP), and our metatheory is
incompatible with homotopy type theory \cite{HoTTbook}. Our notation
on paper is close to Agda's: we write $\Pi$ types as $(x:A)\ra B$,
$\Sigma$ types as $(x:A)\times B$, we use implicit arguments and
overloading extensively, for example instantiation $\blank[\blank]$ is
overloaded for types and terms. For isomorphisms, we use the notation
$f:X\cong Y:g$ meaning $f:X\ra Y$ and $g :Y\ra X$ together with a
$\beta$ equality $f\,(g\,y) = y$ for all $y$, and an $\eta$ equality
saying $g\,(f\,x) = x$ for all $x$ (we think of $f$ as a destructor
and $g$ as a constructor). Following Voevodsky \cite{voevodsky2015csystemdefineduniversecategory}, we call
a proof relevant theorem a \emph{problem} and its proof a \emph{construction}.
The words theorem, lemma, proof refer to propositions.

\vspace{-0.5em}
\section{Single substitution syntax}
\label{sec:tt}

In this section, we introduce the syntax of type theory with function
space and universes in a minimalistic way, only introducing operations
that are unavoidable. We eschew boilerplate by only defining
well-formed, well-scoped, well-typed terms that are quotiented by
conversion. This means that the equalities that hold between terms are
the ways we can convert terms into each other when running them as
programs. We give a tutorial-style introduction to the syntax, we do
not assume prior knowledge of the metatheory of type theory.

First of all, we need a sort of terms which have to be indexed by
types because we only want well-typed terms. Both types and terms can
include variables, and to keep track of the currently available
variables, we also index them by contexts: a context is a list of
types, the length is the number of available variables, and we add new
variables at the end (that is, they are snoc-lists rather than
cons-lists). Types are also indexed by their universe level, this is a
technical requirement for avoiding Russell's paradox. The index $i$ is
an implicit argument of $\Tm$.
\begin{equation*}
\Con : \Set \hspace{3em} 
\Ty : \Con\ra\N\ra\Set \hspace{3em} 
\Tm : (\Gamma:\Con)\ra\Ty\,\Gamma\,i\ra\Set
\end{equation*}
Just as lists have two constructors, there are two ways to form
contexts: the empty context $\diamond$ and context extension $\ext$
which is like the snoc operation for lists. Context extension takes a
type which can have variables in the context preceding the type.
\[
\diamond : \Con \hspace{3em}
\blank\ext\blank : (\Gamma:\Con)\ra\Ty\,\Gamma\,i\ra\Con
\]
We will refer to variables by their distance from the end of the
context. We would like to describe the operation providing the last
variable in a context, first of all we need an operation providing the last
variable in the context. This is the zero De Bruijn index \cite{DEBRUIJN1972381}, which we
denote by $\q : \Tm\,(\Gamma\ext A)\,A'$. It is a term in an
extended context, and it should have type $A$, but the issue is that
$A : \Ty\,\Gamma$, and $A' : \Ty\,(\Gamma\ext A)$. We need a way to
weaken the type $A$ to obtain such an $A'$. For this, we will
introduce an operation $\blank[\p] : \Ty\,\Gamma\,i\ra\Ty\,(\Gamma\ext
A)\,i$ and define $A' := A[\p]$. Instead of introducing just
$\blank[\p]$, we generalise a bit and add a new sort $\Sub$. For
now, $\Sub$ only has the single element $\p$ and we introduce the operation
$\blank[\blank]$ called \emph{instantiation}. Elements of $\Sub$
will later be called substitutions, hence the name, but for now, we
only have single end-of context weakenings in $\Sub$. The context
$\Delta$ is called the domain, $\Gamma$ the codomain of a
$\Sub\,\Delta\,\Gamma$.
\begin{alignat*}{10}
  & \Sub && : \Con\ra\Con\ra\Set &&\hspace{1.5em}
  & \blank[\blank] && : \Ty\,\Gamma\,i\ra\Sub\,\Delta\,\Gamma\ra\Ty\,\Delta\,i \hspace{1.5em}
  \p && : \Sub\,(\Gamma\ext A)\,\Gamma \hspace{1.5em}
  && \q && : \Tm\,(\Gamma\ext A)\,(A[\p])
\end{alignat*}
The operations $\p$ and $\q$ take three implicit parameters, $\Gamma$,
$i$ and $A$, while $\blank[\blank]$ takes $\Gamma$, $i$ and $\Delta$
implicitly. Still, we only have the last variable $\q$ in the context,
we don't have e.g.\ the last but one variable in $\Tm\,(\Gamma\ext
A\ext B)\,(A[\p][\p])$ where we had to weaken $A$ twice to make it fit
with its context. To obtain more variables, we also allow weakening of
terms, more precisely, we introduce an instantiation operation for
terms with an overloaded notation.
\begin{alignat*}{10}
& \blank[\blank] && : \Tm\,\Gamma\,A\ra(\gamma:\Sub\,\Delta\,\Gamma)\ra\Tm\,\Delta\,(A[\gamma])
\end{alignat*}
Note that this is a dependent function as the type of the resulting
term has to be weakened the same way as the term itself. Now we can
define all variables counting from the end of the context (De Bruijn
indices): $0 := \q$, $1 := \q[\p]$, $2 := \q[\p][\p]$,
$3 := \q[\p][\p][\p]$, and so on.

Next we add dependent function space $\Pi$: the domain of a dependent
function is a type in some context $\Gamma$, and the codomain can also
refer to a variable in the domain, so we extend the context of the
codomain type with the type of the domain. For simplicity, both types
are at the same level (this can be remedied using $\Lift$, see later).
\begin{alignat*}{10}
& \Pi && : (A:\Ty\,\Gamma\,i)\ra\Ty\,(\Gamma\ext A)\,i\ra\Ty\,\Gamma\,i
\end{alignat*}
The nondependent function space $\blank\Ra\blank :
\Ty\,\Gamma\,i\ra\Ty\,\Gamma\,i\ra\Ty\,\Gamma\,i$ is a special case of $\Pi$
and we define it as the abbreviation $A\Ra B := \Pi\,A\,(B[\p])$.

Because we introduced weakening, we now have to explain what happens
to a $\Pi$ type once we weaken it. Assume $A : \Ty\,\Gamma\,i$,
$B:\Ty\,(\Gamma\ext A)\,i$ and we have $\p:\Sub\,(\Gamma\ext
C)\,\Gamma$. Then we say $(\Pi\,A\,B)[\p] = \Pi\,(A[\p])\,(B[\p'])$,
but $B$ cannot be weakened by $\p$ because there we need $\p' :
\Sub\,(\Gamma\ext C\ext A[\p])$ $(\Gamma\ext A)$. So $p'$ has to be a % hack
new kind of weakening which adds a new variable in the last but one
position. But after introducing $\p'$, we would need another equation computing $(\Pi\,A\,B)[\p']$,
introducing a new weakening which adds a new variable in the last but
two position, and so on. We solve all of these issues by allowing the
\emph{lifting} of a weakening: $\gamma^+$ will be the same as the
weakening $\gamma : \Sub\,\Delta\,\Gamma$, except that it will add an
extra variable both at the end of the domain and codomain
context.
\begin{alignat*}{10}
& \blank^+ && : (\gamma:\Sub\,\Delta\,\Gamma)\ra\Sub\,(\Delta\ext A[\gamma])\,(\Gamma\ext A)
\end{alignat*}
Note that $A$ is an implicit argument of $\blank^+$, and that it has
to be instantiated by the weakening $\gamma$ in the domain context in
order to fit in. Now we explain how instantiation acts on $\Pi$ by
the following equation.
\begin{alignat*}{10}
& \Pi[] && : (\Pi\,A\,B)[\gamma] = \Pi\,(A[\gamma])\,(B[\gamma^+])
\end{alignat*}
The above equation has 6 implicit arguments, namely $\Gamma$, $i$,
$A$, $B$, $\Delta$ and $\gamma$.  This rule works for any weakening,
no matter how many $\blank^+$s have been applied to it. As of now, all
elements of $\Sub$ have the form
$
\p^{+^n} : \Sub\,(\Gamma\ext A\ext B_1[\p]\ext B_2[\p^+]\ext\dots\ext B_n[\p^{+^{n-1}}])\,(\Gamma\ext B_1\ext B_2\ext\dots\ext B_n),
$
where $^{+^n}$ means the $n$-times iteration of $\blank^+$.

Now that we have weakenings of the form $\gamma^+$, we have to say how
they act on variables, that is, terms of the form
$\q[\p]\dots[\p]$. We express this using two rules: we say what
$\q[\gamma^+]$ computes to, and what $b[\p][\gamma^+]$ computes to
where $b$ is an arbitrary term. $\q[\gamma^+]$ should be the same as
$\q$ (with different implicit arguments as the $\q$ in
$\q[\gamma^+]$), as the weakening happens somewhere in the middle of
the context, so the index of the variable remains unchanged. However,
assuming $\q:\Tm\,(\Gamma\ext A)\,(A[\p])$, we have $\q[\gamma^+] :
\Tm\,(\Delta\ext A[\gamma])\,(A[\p][\gamma^+])$, but in the same
context, we have $\q : \Tm\,(\Delta\ext A[\gamma])\,(A[\gamma][\p])$,
hence the terms in the two sides of the equation $\q[\gamma^+] = \q$
have different types. But these two types should be the same:
weakening a type at the end of the context, and then applying another
lifted weakening should be the same as first weakening somewhere and
then weakening at the end. We first assume this equation for types,
then the equation $\q[\gamma^+] = \q$ becomes well-typed (in the
metatheory). The equation for $b[\p][\gamma^+]$ has the same shape as
the newly assumed rule for types. Thus we add the following three
equations.
\[
[\p][^+] : B[\p][\gamma^+] = B[\gamma][\p] \hspace{3em}
[\p][^+] : b[\p][\gamma^+] = b[\gamma][\p] \hspace{3em}
\q[^+] : \q[\gamma^+] = \q
\]
The second and third equations only make sense because of the first
equation. The phenomenon that later equations/operations depend on
previous equations is common in GATs.

Now that we have $[\p][^+]$ for types, we can derive its instantiation rule by
$
(A\Ra B)[\gamma] =
\Pi\,A\,(B[\p])[\gamma] \overset{\Pi[]}{=}
\Pi\,(A[\gamma])\,(B[\p][\gamma^+]) \overset{[\p][^+]}{=}
\Pi\,(A[\gamma])\,(B[\gamma][\p]) =
(A[\gamma])\Ra(B[\gamma]).
$

So far our only terms are variables, but we would like to define
functions via lambda abstraction. Abstraction takes a term in a
context extended by the domain of the function. It comes with a rule
for instantiation analogous to $\Pi[]$.
\[
\lam : \Tm\,(\Gamma\ext A)\,B\ra\Tm\,\Gamma\,(\Pi\,A\,B) \hspace{3em}
\lam[] : (\lam\,b)[\gamma] = \lam\,(b[\gamma^+])
\]
Note that the equation $\lam[]$ only makes sense because of the
previous equation $\Pi[]$: the left hand side is in
$\Tm\,\Delta\,(\Pi\,A\,B[\gamma])$, in the right hand side,
$b[\gamma^+] : \Tm\,(\Delta\ext A[\gamma])\,(B[\gamma^+])$, hence the
right hand side is in
$\Tm\,\Delta\,(\Pi\,(A[\gamma])\,(B[\gamma^+]))$.

The functions $\lambda x\ldotp x$ and $\lambda x\,y\ldotp x$ can be defined in our
syntax as $\lam\,\q : \Tm\,\Gamma\,(A\Ra A)$ and
$\lam\,(\lam\,(\q[\p])) :
\Tm\,\Gamma\,\big(\Pi\,A$ $(B\Ra(A[\p]))\big)$ which make sense for all
$\Gamma$, $A$, $B$.

The application operation for dependent function space is a bit
tricky, because the return type depends on the input: the argument of
the function will appear in the return type. We write application by
an infix $\blank\cdot\blank :
\Tm\,\Gamma\,(\Pi\,A\,B)\ra(a:\Tm\,\Gamma\,A)\ra\Tm\,\Gamma\,B'$,
where $B' : \Ty\,\Gamma$ should be $B:\Ty\,(\Gamma\ext A)$ where the
last variable is \emph{substituted} (instantiated) by $a$. For this,
we introduce a new element of $\Sub$ called single substitution which
goes the opposite way of $\p$. Now we can use instantiation
$\blank[\blank]$ to substitute the last variable in $B$.
\[
\langle\blank\rangle : \Tm\,\Gamma\,A\ra\Sub\,\Gamma\,(\Gamma\ext A) \hspace{3em}
\blank\cdot\blank : \Tm\,\Gamma\,(\Pi\,A\,B)\ra(a:\Tm\,\Gamma\,A)\ra\Tm\,\Gamma\,(B[\langle a\rangle])
\]
We could have introduced new, separate sorts and $\blank[\blank]$
operations for weakenings and substitutions, but we merge them for
simplicity. There is no need for separation: just as weakenings can be
lifted, single substitutions can also be lifted, and the
weakening-rules $\Pi[]$, $\lam[]$ also work for single
substitutions. We won't have more ways to introduce elements of
$\Sub$, and there are no equations on $\Sub$,
as the equations relating substitutions are defined for their action
on types and terms. An element of $\Sub$ is
either a single weakening $\p$ lifted a finite number of times, or a single
substitution lifted a finite number of times. We call elements of $\Sub$
substitutions for simplicity.

With the introduction of $\blank\cdot\blank$, we need a new
substitution rule, but again it only makes sense with an extra
equation on types saying that first substituting the last variable and
then an arbitrary substitution is the same as first the lifted version
of the arbitrary substitution that does not touch the last variable,
and then substituting the last variable.
\[
[\langle\rangle][] : B[\langle a\rangle][\gamma] = B[\gamma^+][\langle a[\gamma]\rangle] \hspace{2em}
{\cdot}[] : (t\cdot a)[\gamma] = (t[\gamma])\cdot(a[\gamma])
\]
Following the introduction of the operation $\langle\blank\rangle$, we
need to explain how it acts on variables. Given a variable in the
middle of the context (a term $b[\p]$), substituting its last variable
simply returns $b$. Substituting the last variable $\q$ reads out the
term from $\langle\blank\rangle$. As usual, the rules only typecheck
if we have an equation for types.
\[
[\p][\langle\rangle] : B[\p][\langle a\rangle] = B \hspace{3em}
[\p][\langle\rangle] : b[\p ][\langle a\rangle] = b \hspace{3em}
\q[\langle\rangle] : \q[\langle a\rangle] = a
\]
We don't have a separate sort of variables, so the rule
$[\p][\langle\rangle]$ holds not only for variables, but for arbitrary
terms. This is not an issue, as weakening a term at the end of its
context, and then substituting the newly introduced variable is the
same as not doing anything.

We add the computation and uniqueness rule for function space, where
the uniqueness rule again needs an extra equation on the codomain type
of the function to make sense (we mentioned the term version of this equation in the introduction).
\[
\Pi\beta : \lam\,b\cdot a = b[\langle a\rangle] \hspace{3em}
[\p^+][\langle\q\rangle] : B[\p^+][\langle\q\rangle] = B \hspace{3em}
\Pi\eta : t = \lam\,(t[\p]\cdot\q)
\]

We add the rules for universes. A universe is a type containing codes
for types, this is witnessed by $\El$ (elements) and $\cd$ (code),
which make an isomorphism between terms of type $\U\,i$ and types of
level $i$. All three operations come with substitution rules.
% ($\El[]$
% and $\cd[]$ are interderivable, so it would be enough to assume one of
% the two, but we add both to completely align with the algorithm
% \cite{DBLP:conf/fscd/KaposiX24} generating the rules from the SOGAT
% (\ref{eq:tt})).
% \begin{alignat*}{10}
%   & \U && : (i:\N)\ra\Ty\,\Gamma\,(1+i)\hspace{3em} && \U[] && : (\U\,i)[\gamma] = \U\,i                           && \U\beta && : \El\,(\cd\,A) = A \\        
%   & \El && : \Tm\,\Gamma\,(\U\,i) \ra \Ty\,\Gamma\,i && \El[] && : (\El\,\hat{A})[\gamma] = \El\,(\hat{A}[\gamma]) \hspace{3em} && \U\eta && : \cd\,(\El\,\hat{A}) = \hat{A} \\
%   & \cd && : \Ty\,\Gamma\,i\ra\Tm\,\Gamma\,(\U\,i) && \cd[] && : (\cd\,A)[\gamma] = \cd\,(A[\gamma])
% \end{alignat*}
\begin{alignat*}{10}
  & \U && : (i:\N)\ra\Ty\,\Gamma\,(1+i) \hspace{1.5em} && \El && : \Tm\,\Gamma\,(\U\,i) \ra \Ty\,\Gamma\,i                          && \cd && : \Ty\,\Gamma\,i\ra\Tm\,\Gamma\,(\U\,i) \hspace{1.5em} && \U\beta && : \El\,(\cd\,A) = A \\
  & \U[] && : (\U\,i)[\gamma] = \U\,i                  && \El[] && : (\El\,\hat{A})[\gamma] = \El\,(\hat{A}[\gamma]) \hspace{1.5em} && \cd[] && : (\cd\,A)[\gamma] = \cd\,(A[\gamma])                && \U\eta && : \cd\,(\El\,\hat{A}) = \hat{A}
\end{alignat*}
Finally, we add the rules for moving types one level up. Because $\Pi$
needs two types at the same level, without lifting, we cannot even
define the polymorphic identity function.
% \begin{alignat*}{10}
%   & \Lift && : \Ty\,\Gamma\,i\ra\Ty\,\Gamma\,(1+i)\hspace{3em} && (\Lift\,A)[\gamma] = \Lift\,(A[\gamma])\hspace{3em} && \Lift\beta && : \un\,(\mk\,a) = a \\
%   & \mk && : \Tm\,\Gamma\,A \ra \Tm\,\Gamma\,(\Lift\,A) && (\mk\,a)[\gamma] = \mk\,(a[\gamma])              && \Lift\eta && : \mk\,(\un\,a) = a \\
%   & \un && : \Tm\,\Gamma\,(\Lift\,A) \ra \Tm\,\Gamma\,A && (\un\,a)[\gamma] = \un\,(a[\gamma])
% \end{alignat*}
\begin{alignat*}{10}
  & \Lift : \Ty\,\Gamma\,i\ra\Ty\,\Gamma\,(1+i) \hspace{.9em} && \mk : \Tm\,\Gamma\,A \ra \Tm\,\Gamma\,(\Lift\,A) \hspace{.9em} && \un : \Tm\,\Gamma\,(\Lift\,A) \ra \Tm\,\Gamma\,A \hspace{.9em} && \un\,(\mk\,a) = a \\
  & (\Lift\,A)[\gamma] = \Lift\,(A[\gamma])                   && (\mk\,a)[\gamma] = \mk\,(a[\gamma])                            && (\un\,a)[\gamma] = \un\,(a[\gamma])                            && \mk\,(\un\,a) = a
\end{alignat*}
Note that we don't have definitional commutation of type formers and
$\Lift$ such as $\Lift\,(A\Ra B) = (\Lift\,A\Ra\Lift\,B)$. It however
holds as a definitional isomorphism (see Subsection
\ref{sec:examples}). This adds extra administration when using our
theory, but makes stating our theory simpler, and includes more
models. The situation is similar to the correspondence between
$\Ty\,\Gamma\,i$ and $\Tm\,\Gamma\,(\U\,i)$ which could have been
stated as an equality (making universes {\`a} la Russell).

\begin{definition}[Single substitution calculus with $\Pi$, $\U$, $\Lift$ (\href{https://szumixie.github.io/single-subst/TT.SSC.Syntax.html}{\Agda})]\label{def:tt}
This concludes the definition of our basic type theory. For
reference, we list the same rules again in Appendix \ref{app:tt}.
\end{definition}

As a sanity check for our notion of SSC-model, we defined the standard (metacircular\slash set/type) model \cite{DBLP:conf/popl/AltenkirchK16} of the
SSC-based calculus in Agda (\href{https://szumixie.github.io/single-subst/TT.SSC.Standard.html}{\Agda}).

The theory can be extended with new type and term formers in the same
fashion: each operation has to be indexed over contexts and come
 with a substitution rule. There is no need to add more
structural (substitution calculus) rules.\footnote{In \cite{DBLP:conf/fscd/KaposiX24},
the same structural rules suffice to describe the first-order theory for any SOGAT.}
For example, we show how to
extend Definition \ref{def:tt} with $\top$ and $\Sigma$ types.
\begin{definition}[Single substitution calculus with $\Pi$, $\U$, $\Lift$, $\top$, and $\Sigma$ types]\label{def:sigma}
  We extend Definition \ref{def:tt} with the following rules.
  \begin{alignat*}{10}
    & \rlap{$
  \top : \Ty\,\Gamma\,0 \hspace{3.9em}
  \top[] : \top[\gamma] = \top\hspace{3.9em}
  \top\eta : t = \tt \hspace{3.9em}
  \tt : \Tm\,\Gamma\,\top \hspace{3.9em}
  \tt[] : \tt[\gamma] = \tt
      $} \\[0.5em]
    & \Sigma && : (A:\Ty\,\Gamma\,i)\ra\Ty\,(\Gamma\ext A)\,i\ra\Ty\,\Gamma\,i && \Sigma[] && : (\Sigma\,A\,B)[\gamma] = \Sigma\,(A[\gamma])\,(B[\gamma^+]) \\
    & \fst,\snd && : \Tm\,\Gamma\,(\Sigma\,A\,B)\cong(a:\Tm\,\Gamma\,A)\times\Tm\,\Gamma\,(B[\langle a\rangle]) : \blank{,}\blank \hspace{3.2em} && {,}[] && : (a,b)[\gamma] = (a[\gamma],b[\gamma])
  \end{alignat*}
  To save space, we compressed the introduction, elimination, $\beta$
  and $\eta$ rules of $\Sigma$ into an isomorphism, and did not list
  substitution law for $\fst$, $\snd$ as they are derivable (see
  below).
\end{definition}

\subsection{Examples}\label{sec:examples}

In this subsection, we show how to use the above defined calculus, and
illustrate its deficiencies: some important equations are not
derivable in any model, only admissible in the syntax. These equations
will be proven in Section \ref{sec:admissible}.

The polymorphic identity function for types at the level 0 is defined
as
\[
\lam\,(\lam\,\q) : \Tm\,\diamond\,(\Pi\,(\U\,0)\,(\Lift\,(\El\,\q)\Ra\Lift\,(\El\,\q))).
\]
Note that readability of this term essentially relies on implicit
arguments. For example, $\q$ takes $(\diamond\ext\U\,0)$ as its first,
$1$ as its second and $\Lift\,(\El\,\q)$ as its third implicit
argument, and these arguments themselves are terms written using
implicit arguments. Intrinsically typed terms are the same as
derivation trees, we illustrate this by deriving this term as follows.
\[
\infer[\text{(*)}]{\lam\,(\lam\,\q) : \Tm\,\diamond\,\big(\Pi\,(\U\,0)\,(\Lift\,(\El\,\q)\Ra\Lift\,(\El\,\q))\big)}{\infer{\lam\,\q : \Tm\,(\diamond\ext\U\,0)\,(\Lift\,(\El\,\q)\Ra\Lift\,(\El\,\q))}{\infer{\q : \Tm\,(\diamond\ext\U\,0\ext\Lift\,(\El\,\q))\,(\Lift\,(\El\,\q)[\p])}{}}}
\]
Note that in step (*), $\lam$ made it sure that $\U\,0$ and
$\Lift\,(\El\,\q)\Ra\Lift\,(\El\,\q)$ are types at the same level,
which is why we had to lift $\El\,\q$. Another subtlety is
happening when deriving the third implicit argument of $\q$, because we have
to coerce $\q$ along the equality $\U[]$:
\[
\infer{\Lift\,(\El\,\q) : \Ty\,(\diamond\ext\U\,0)\,1}{\infer{\El\,\q : \Ty\,(\diamond\ext\U\,0)\,0}{\infer{\q : \Tm\,(\diamond\ext\U\,0)\,(\U\,0)}{\infer{\q : \Tm\,(\diamond\ext\U\,0)\,((\U\,0)[\p])}{} && \infer{\U[] : (\U\,0)[\p] = \U\,0}{}}}}
\]
When working informally (or in extensional type theory), we don't
write coercions.%, but Agda requires it.

Given a type $A:\Ty\,\diamond\,0$, and $a : \Tm\,\diamond\,A$, we have
$
\lam\,(\lam\,\q)\cdot\cd\,A\cdot a \overset{\Pi\beta}{=}
\lam\,\q[\langle\cd\,A\rangle]\cdot a \overset{\lam[]}{=}
\lam\,(\q[\langle\cd\,A\rangle^+])\cdot a\, \overset{[\q][^+]}{=}
\lam\,\q\cdot a \overset{\Pi\beta}{=}
\q[\langle a\rangle] \overset{\q[\langle\rangle]}{=}
a.
$
If we want to specify open inputs $A':\Ty\,(\diamond\ext C)\,0$, and
$a' : \Tm\,(\diamond\ext C)\,A'$, we have to weaken our function to
accept them:
$
\lam\,(\lam\,\q)[\p] : \Tm\,(\diamond\ext C)\,(\Pi\,(\U\,0)\,(\Lift\,(\El\,\q)\Ra\Lift\,(\El\,\q))[\p]).
$
But we have
$
\lam\,(\lam\,\q)[\p] \overset{\lam[]\,2\text{x}}{=} \lam\,(\lam\,(\q[\p^{++}])) \overset{\q[^+]}{=} \lam\,(\lam\,\q)
$
(with different implicit arguments for the $\lam$s on the left and right hand side), so we are able to apply $\cd\,A'$
and $a'$ just as before. What we cannot derive is to directly weaken a
term in context $\diamond$ to an arbitrary context $\Gamma$, where
$\Gamma$ is a metavariable. If we know that $\Gamma$ has length $n$,
then we can do the weakening by applying $\blank[\p]$ to it $n$ times.
If we work
in the syntax, then we can derive such weakenings (and more) using the
methods of Section \ref{sec:admissible}.

The operations for $\U$ and $\Lift$ can be written more concisely as
the following isomorphisms:
\[
\El:\Tm\,\Gamma\,(\U\,i) \cong \Ty\,\Gamma\,i:\cd \hspace{5em} \un:\Tm\,\Gamma\,(\Lift\,A)\cong\Tm\,\Gamma\,A:\mk
\]
The $\U\beta$, $\U\eta$ and $\Lift\beta$, $\Lift\eta$ rules express
that the round-trips are identities. These isomorphisms are also
\emph{natural}, or compatible with substitutions: these are expressed
by the equations $\El[]$, $\un[]$, $\cd[]$, $\mk[]$.

$\Sigma$ types were introduced by a similar isomorphism, and we give
names to the round-trips equations according to our convention introduced in
Section \ref{sec:metatheory}: $\Sigma\beta_1 : \fst\,(a,b) = a$,
$\Sigma\beta_2 : \snd\,(a,b) = b$ and $\Sigma\eta : w =
(\fst\,w,\snd\,w)$. Naturality for $\Sigma$ was only stated in one
direction (${,}[]$), naturalities in the other direction are
derivable, for $\fst$ this is
$
(\fst\,w)[\gamma]  \overset{\Sigma\beta_1}{=}                          
\fst\,\big((\fst\,w)[\gamma],(\snd\,w)[\gamma]\big) \overset{{,}[]}{=} 
\fst\,\big((\fst\,w,\snd\,w)[\gamma]\big)  \overset{\Sigma\eta}=        
\fst\,(w[\gamma]).
$
For $\Pi$ types, we construct a similar isomorphism
$
(\lambda t\ldotp t[\p]\cdot\q) : \Tm\,\Gamma\,(\Pi\,A\,B) \cong \Tm\,(\Gamma\ext A)\,B : \lam.
$
The first round-trip equality is proven as
$\beta : (\lam\,t)[\p]\cdot\q \overset{\lam[]}{=} \lam\,(t[\p^+])\cdot\q \overset{\Pi\beta}{=} t[p^+][\langle\q\rangle] \overset{\text{(\ref{eq:lifted4})}}{=} t$.
The last step is the admissible equation (\ref{eq:lifted4}) from
Section \ref{sec:admissible}. As above with weakening into an
arbitrary context, this equation holds in the syntax, but not in an
arbitrary model.
%and also if $t$ is not a metavariable, but an operation, e.g. if $t = \q$ or $t = \q[\p]\cdot\q$).
The other round-trip is exactly $\Pi\eta : \lam\,(t[\p]\cdot\q) = t$.

We construct another isomorphism between terms depending on a variable and
terms depending on the lifted version of the same variable:
\begin{equation}\label{eq:liftvar}
(\lambda t\ldotp t[\p^+][\langle\un\,\q\rangle]) : \Tm\,(\Gamma\ext A)\,B \cong \Tm\,(\Gamma\ext\Lift\,A)\,(B[\p^+][\langle\q\rangle]) : (\lambda t\ldotp t[\p^+][\langle\mk\,\q\rangle])
\end{equation}
To derive the round-trips, we again need admissible equations. We prove
the first round-trip equality in Appendix \ref{app:tt}, the other
is analogous.
Once we have all the above isomorphisms, we can compose them to obtain commutation of $\Lift$ and $\Pi$:
\begin{alignat*}{10}
  & \Tm\,\Gamma\,\big(\Lift\,(\Pi\,A\,B)\big) \overset{\text{$\Lift$}}{\cong}
  \Tm\,\Gamma\,(\Pi\,A\,B) \overset{\text{$\Pi$}}{\cong} 
  \Tm\,(\Gamma\ext A)\,B \overset{\text{$\Lift$}}{\cong} 
  \Tm\,(\Gamma\ext A)\,(\Lift\,B) \overset{\text{lifted var}}{\cong} \\
  & \Tm\,(\Gamma\ext \Lift\,A)\,(\Lift\,B[\p^+][\langle\un\,\q\rangle]) \overset{\text{$\Pi$}}{\cong}
  \Tm\,\Gamma\,\big(\Pi\,(\Lift\,A)\,(\Lift\,B[\p^+][\langle\un\,\q\rangle])\big)
\end{alignat*}
We hope that the above examples demonstrated the need for equations
(\ref{eq:lifted1})--(\ref{eq:lifted4}) which we will prove in the next
section.

% \subsection{Standard model}
% \label{sec:standard}
% 
% As a sanity check for our notion of SSC-model, we defined the standard
% (metacircular \cite{DBLP:conf/popl/AltenkirchK16}) model of the
% SSC-based calculus in Agda. For convenience, we used an
% inductive-recursive universe to define the infinite hierarchy of types
% \cite{DBLP:journals/jsyml/Dybjer00,DBLP:conf/csl/Kovacs22}. All
% SSC-equations hold by reflexivity in this model. We defined the
% analogous CwF-based standard model.

\section{Admissible equations}
\label{sec:admissible}

In this section we show that our SSC-based syntax is isomorphic to the
CwF-based syntax with the same type formers.\footnote{Using the two translations from SOGATs to GATs
\cite{DBLP:conf/fscd/KaposiX24}, we can rephrase the contents of this
section as follows: the parallel and single substitution syntax of the SOGAT (\ref{eq:tt}) are
isomorphic.}

\begin{definition}[CwF (see \href{https://szumixie.github.io/single-subst/TT.CwF.Syntax.html}{\Agda})]\label{def:cwf}
  A \emph{category with $\N$-many families} is defined as a category (objects
   $\Con$, morphisms $\Sub$) with a terminal object $\diamond$
  ($\epsilon$ denotes the unique morphism into it); for each $i$ a
  presheaf of types (action on objects denoted $\Ty\,\blank\,i$,
  action on morphisms $\blank[\blank]$); for each $i$ a locally
  representable dependent presheaf $\Tm$ over $\Ty\,\blank\,i$
  (actions denoted $\Tm$, $\blank[\blank]$, local representability is
  denoted $\blank\ext\blank : (\Gamma:\Con)\ra\Ty\,\Gamma\,i\ra\Con$
  with an isomorphism
  $
  (\p\circ\blank,\q[\blank]) : \Sub\,\Delta\,(\Gamma\ext A) \cong (\gamma:\Sub\,\Delta\,\Gamma)\times\Tm\,\Delta\,(A[\gamma]) : (\blank,\blank)
  $
  natural in $\Delta$).
\end{definition}
In the rest of this section, we work with the SSC syntax denoted
$\Con$, $\Sub$, $\Ty$, $\Tm$, and so on, towards the goal of defining
all CwF operations and equations.

The following four equations correspond to equations in B-systems which cannot
be derived from the equations of SSC, these are versions of equations in Section
\ref{sec:tt} lifted over arbitrary amount of $\blank^+$s.
\setlength{\columnsep}{4em}
\begin{multicols}{2}\noindent
\begin{alignat}{10}
  B[\p^{+^n}][(\gamma^+)^{+^n}] & = B[\gamma^{+^n}][\p^{+^n}] \label{eq:lifted1} \\
  B[\p^{+^n}][\langle a \rangle^{+^n}] & = B \label{eq:lifted2}
\end{alignat}
\columnbreak
\begin{alignat}{10}
  \hspace{-4em} B[\langle a \rangle^{+^n}][\gamma^{+^n}] & = B[(\gamma^+)^{+^n}][\langle a[\gamma] \rangle^{+^n}] \label{eq:lifted3} \\
  \hspace{-4em} B[(\p^+)^{+^n}][\langle\q\rangle^{+^n}] & = B\label{eq:lifted4}
\end{alignat}
\end{multicols}
\vspace{-2em}
The corresponding equations on terms are also needed. Note that we do not have
the non-lifted \emph{term} versions of the last two equations in SSC.

To formally state equations (\ref{eq:lifted1})--(\ref{eq:lifted4}), we
define telescopes as an inductive type together with a recursive
operation to append telescopes to contexts (\href{https://szumixie.github.io/single-subst/TT.SSC.Tel.html#Tel}{\Agda}).
\begin{alignat*}{10}
  & \Tel && : \Con\ra\Set   && \rlap{$\blank+\blank : (\Gamma:\Con)\ra\Tel\,\Gamma\ra\Con$} \\
  & \diamond && : \Tel\,\Gamma && \Gamma+\diamond && := \Gamma \\
  & \blank\ext\blank && : (\Omega:\Tel\,\Gamma)\ra\Ty\,(\Gamma+\Omega)\ra\Tel\,\Gamma \hspace{3em} && \Gamma+(\Omega\ext A) && := (\Gamma+\Omega)\ext A \hspace{3em}
\end{alignat*}
We then define a lifting operation over any telescope mutually with an
instantiation operation on telescopes by recursion on telescopes (\href{https://szumixie.github.io/single-subst/TT.SSC.Tel.html#_[_]\%E1\%B5\%80\%CB\%A1}{\Agda}).
\[
\blank[\blank] : \Tel\,\Gamma\ra\Sub\,\Delta\,\Gamma\ra\Tel\,\Delta \hspace{3em}
\blank^{+^\Omega} : (\gamma:\Sub\,\Delta\,\Gamma)\ra\Sub\,(\Delta+\Omega[\gamma])\,(\Gamma+\Omega)
\]

If we try to prove the equations in (\ref{eq:lifted1})--(\ref{eq:lifted4}) by induction on the types
and terms, we run into difficulties in the instantiation case, as we would have
to commute the instantiations without the equations to do so. Therefore, we
first $\alpha$-normalize the syntax to compute away all the instantiations, then
prove the equations by induction on the $\alpha$-normal forms.

A type or term is \emph{$\alpha$-normal} if it does not contain instantiation
operations, except at the leaves of the syntax tree as variables. However,
$\alpha$-normal terms can still contain $\beta$/$\eta$ redexes, so that we do
not need to do full normalization. We define variables and $\alpha$-normal forms
as inductive predicates in $\Prop$ as follows (\href{https://szumixie.github.io/single-subst/TT.SSC.AlphaNorm.html#Var}{\Agda}).
\begin{itemize}
  \item $\q:\Tm\,(\Gamma\ext A)\ (A[\p])$ is a variable.
  \item $x[\p]:\Tm\,(\Gamma\ext A)\ (B[\p])$ is a variable if $x:\Tm\,\Gamma\,B$
  is a variable.
  \item $x:\Tm\,\Gamma\,A$ is $\alpha$-normal if $x$ is a variable.
  \item $\Pi\,A\,B:\Ty\,\Gamma\,i$ is $\alpha$-normal if $A$ and $B$ are
  $\alpha$-normal types.
  \item $f\cdot a:\Tm\,\Gamma\,(B[\langle a \rangle])$ is $\alpha$-normal if
  $A$ and $B$ are $\alpha$-normal types, $f:\Tm\,\Gamma\,(\Pi\,A\,B)$ and
  $a:\Tm\,\Gamma\,A$ are $\alpha$-normal terms.
\end{itemize}
The predicate is defined similarly for the other type and term formers such as
$\lam$. Notably, we do not state that instantiated types/terms are
$\alpha$-normal (except variables). We truncate the
$\alpha$-normal predicate to be a proposition to allow interpreting into it from the
quotiented syntax. Without truncation, interpretation into $\alpha$-normal forms
would expose differences between $\beta$/$\eta$-equal terms (and such an interpretation is not definable).

We prove that the $\alpha$-normal predicate holds for $\alpha$-normal types and
terms instantiated with $\alpha$-normal substitutions (\href{https://szumixie.github.io/single-subst/TT.SSC.AlphaNorm.html\#[]\%E1\%B4\%BA\%E1\%B4\%BE}{\Agda}). However, this needs to be
proved separately for weakenings and $\alpha$-normal single substitutions for
the induction to be structural. We define predicates for weakening (\href{https://szumixie.github.io/single-subst/TT.SSC.AlphaNorm.html#Wk}{\Agda}) and
$\alpha$-normal single substitutions (\href{https://szumixie.github.io/single-subst/TT.SSC.AlphaNorm.html#NSSub}{\Agda}) as follows.
\begin{itemize}
  \item $\p$ is a weakening.
  \item $\gamma^+$ is a weakening if $\gamma$ is a weakening.
  \item $\langle a \rangle$ is an $\alpha$-normal single substitution if $a$ is an
  $\alpha$-normal term.
  \item $\gamma^+$ is an $\alpha$-normal single substitution if $\gamma$ is an
  $\alpha$-normal single substitution.
\end{itemize}
We define $\alpha$-normal substitutions to be the disjoint union of weakenings
and $\alpha$-normal single substitutions (\href{https://szumixie.github.io/single-subst/TT.SSC.AlphaNorm.html#NSub}{\Agda}).

\begin{lemma}[\href{https://szumixie.github.io/single-subst/TT.SSC.AlphaNorm.html\#norm\%E1\%B5\%80}{\Agda}]
  \label{thm:alpha}
  The $\alpha$-normal predicate holds for any type and term.
\end{lemma}
\begin{proof}
  By induction on the syntax, $\alpha$-normalizing the substitutions
  at the same time. Note that induction on the syntax refers to the
  elimination principle of the corresponding QIIT.
\end{proof}

Instead of doing induction on $\alpha$-normal forms for each of equations 
(\ref{eq:lifted1})--(\ref{eq:lifted4}), we define a general lemma which can lift any equation between
instantiations over a telescope. For this we define $\Sub^*$ to be $\Sub$ with
freely added identity and composition operations (\href{https://szumixie.github.io/single-subst/TT.SSC.Path.html#Sub*}{\Agda}). We do not require it to
satisfy the category laws as we will not compare $\Sub^*$s directly for
equality. All
instantiation operations, lifting operations, and substitution rules are
redefined for $\Sub^*$ (\href{https://szumixie.github.io/single-subst/TT.SSC.Path.html}{\Agda}).
\begin{lemma}
  \label{thm:lifting}
  Given $\gamma_0, \gamma_1 : \Sub^*\,\Delta\,\Gamma$, if for any $A:\Ty\,\Gamma$, $A[\gamma_0] = A[\gamma_1]$, and
  for any variable $x:\Tm\,\Gamma\,A$, $x[\gamma_0] = x[\gamma_1]$,
  then:
  \begin{itemize}
    \item \textnormal{(\href{https://szumixie.github.io/single-subst/TT.SSC.Lift.html\#5574}{\Agda})} $\Omega[\gamma_0] = \Omega[\gamma_1]$ for $\Omega:\Tel\,\Gamma$
    \item \textnormal{(\href{https://szumixie.github.io/single-subst/TT.SSC.Lift.html\#5713}{\Agda})} $A[\gamma_0^{+^\Omega}] = A[\gamma_1^{+^\Omega}]$ for $A:\Ty\,(\Gamma+\Omega)$
    \item \textnormal{(\href{https://szumixie.github.io/single-subst/TT.SSC.Lift.html\#5874}{\Agda})} $a[\gamma_0^{+^\Omega}] = a[\gamma_1^{+^\Omega}]$ for $a:\Tm\,(\Gamma+\Omega)\,A$
  \end{itemize}
  Note that the later equations depend on the earlier ones.
\end{lemma}
\begin{proof}
  Assuming the first equation, we prove that $x[\gamma_0^{+^\Omega}] =
  x[\gamma_1^{+^\Omega}]$ for any variable $x$ by induction on the telescope and
  the variable. Then the two latter equations are proven by mutual induction
  on $\alpha$-normalized types and terms, still assuming the first equation.
  Finally we prove the first equation by induction on the telescope, using the
  previously proven equation for types, discharging its assumption.
\end{proof}

We can simulate parallel substitutions using $\Sub^*$ as iterated single
substitutions, however it does not satisfy the equations of parallel
substitutions. Thus we define parallel substitutions as lists of terms. It seems to be
difficult to define the instantiation operation for $\Tms$ directly, so we also define
a map into $\Sub^*$, and reuse the instantiation operation of $\Sub^*$ to
convert it into a sequence of instantiations of single substitutions (\href{https://szumixie.github.io/single-subst/TT.SSC.Parallel.html\#Tms}{\Agda}).
\begin{alignat*}{10}
  & \Tms && : \Con\ra\Con\ra\Con && \llcorner\blank\lrcorner && : \Tms\,\Delta\,\Gamma \to \Sub^*\,\Delta\,\Gamma \\
  & \epsilon && : \Tms\,\Gamma\,\diamond && \llcorner\epsilon\lrcorner && := \id\circ\p\circ\dots\circ\p \\
  & \blank,\blank && : (\gamma:\Tms\,\Delta\,\Gamma)\ra\Tm\,\Delta\,(A[\llcorner\gamma\lrcorner])\ra\Tms\,\Delta\,(\Gamma\ext A)\hspace{1.9em} && \llcorner\gamma,a\lrcorner && := \llcorner\gamma\lrcorner^+\circ\langle a \rangle
\end{alignat*}
The computation of instantiating with $\Tms$ (\href{https://szumixie.github.io/single-subst/TT.SSC.Parallel.html\#559}{\Agda}) is illustrated below. On the
right-hand side, we first weaken $B$ to the context in which the $a_i$'s live,
then replace the variables of $B$ one by one using the $a_i$'s, which are
well-typed because of the previous weakening.
\[
  B[\llcorner\epsilon,a_1,a_2,\dots,a_{n-1},a_n\lrcorner] = 
  B[\p^{+^n}]\dots[\p^{+^n}][\langle a_1 \rangle^{+^{n-1}}][\langle a_2 \rangle^{+^{n-2}}]\dots[\langle a_{n-1} \rangle^+][\langle a_n \rangle]
\]
All CwF operations and equations can be defined with $\Tms$ by induction, using
Lemma \ref{thm:lifting} to avoid further induction on the syntax (\href{https://szumixie.github.io/single-subst/TT.SSC.Parallel.html}{\Agda}).

\begin{problem}[\href{https://szumixie.github.io/single-subst/TT.Isomorphism.html}{\Agda}]
  Contexts, types, and terms in the SSC syntax are isomorphic to the
  corresponding sorts in the CwF syntax. In addition $\Tms$ is isomorphic to CwF
  substitutions.
\end{problem}
\begin{proof}[Construction]\leavevmode
  \begin{description}
    \item[$\Rightarrow$] By recursion on the syntax, SSC operations can be
    trivially interpreted by CwF operations.
    \item[$\Leftarrow$] By recursion on the syntax, using $\Tms$ to interpret
    parallel substitutions.
  \end{description}
  The round-trips are proven by induction on the syntax.
\end{proof}
This also implies that the initial SSC model is the initial CwF.

\section{Minimisation}
\label{sec:minimisation}

In this section we show that if have universes and $\Pi$ types
(which is the case in the theory (\ref{eq:tt})), we can decrease the
number of equations.

In our single substitution calculus, we have equation $[\p][^+]$
stated both for types and terms (here $B$ is a type and $b$ is a term of type $B$):
$B[\p][\gamma^+] = B[\gamma][\p]$ and $b[\p][\gamma^+] =_{[\p][^+]} b[\gamma][\p]$. 
We even need the first equation to typecheck the second one. We made
the dependency explicit by adding a subscript of the equality in the
equation for terms (in Agda, this dependency has to be made
explicit). An alternative presentation of the second equation without
requiring the first one is
$
[\p][^+]' : (e : B[\p][\gamma^+] = B[\gamma][\p])\ra b[\p][\gamma^+] =_{e} b[\gamma][\p].
$
That is, for any type $B$ for which $B[\p][\gamma^+] = B[\gamma][\p]$,
we have the equation (suitably over the input
equation) for any $b$. Thus this is a conditional equation.

It turns out that if we have Coquand-universes, the conditional
$[\p][^+]'$ rule is enough: the input equation holds for $B := \U\,i$
via $\U[]$, thus we get that for any $\hat{A} : \Tm\,\Gamma\,(\U\,i)$,
$\hat{A}[\p][\gamma^+] = \hat{A}[\gamma][\p]$. But every type has a
code in the universe, so for a $B : \Ty\,\Gamma\,i$,
$
B[\p][\gamma^+] \overset{\U\beta}{=}
\El\,(\cd\,B)[\p][\gamma^+] \overset{\El[]}{=}
\El\,((\cd\,B)[\p][\gamma^+]) \overset{[\p][^+]'\,\U[]}{=}
\El\,((\cd\,B)[\gamma][\p]) \overset{\El[]}{=}
\El\,(\cd\,B)[\gamma][\p] \overset{\U\beta}{=}
B[\gamma][\p].
$
Thus, as the input of equation $[\p][^+]'$ holds all the time, we get this equation for every term.
Equation $[\q][^+] : \q[\gamma^+] = \q$ also only makes sense if we
have $[\p][^+]$ for types, so either we have to make $[\q][^+]$
conditional, or we coerce it along the derived equation.

We play the exact same game with $[\p][\langle\rangle]$: we replace it
with the conditional equation
$
[\p][\langle\rangle]' : (e:B[\p][\langle a\rangle] = B)\ra b[\p][\langle a\rangle] =_e b.
$
Notice that we have the input for $B = \U\,i$, and derive the input
equation for all types. We make $\q[\langle\rangle]$ conditional as
well.

Now we turn our attention to the equation
$
  [\langle\rangle][] : B[\langle a\rangle][\gamma] = B[\gamma^+][\langle a[\gamma]\rangle].
$
We were forced to introduce it so that we can state the substitution
rule for function application $\cdot[]$. A conditional version of
$\cdot[]$ is
$
\cdot[]' : (e : B[\langle a\rangle][\gamma] = B[\gamma^+][\langle a[\gamma]\rangle])\ra(t\cdot a)[\gamma] =_e (t[\gamma])\cdot(a[\gamma]).
$
Again, for $B = \U$ we have the assumption $e$. But then $t$ is in
$\Tm\,\Gamma\,(\Pi\,A\,\U)$, which is isomorphic to $\Ty\,(\Gamma\ext
A)$. So we argue as follows for any $B : \Ty\,(\Gamma\ext A)$:
$
  B[\langle a\rangle][\gamma] \overset{\U\beta}{=}
  \El\,(\cd\,(B[\langle a\rangle][\gamma])) \overset{\cd[]}{=}
  \El\,((\cd\,B)[\langle a\rangle][\gamma]) \overset{\Pi\beta}{=}
  \El\,((\lam\,(\cd\,B)\cdot a)[\gamma]) \overset{\cdot[]'\,\U[]}{=}
  \El\,((\lam\,(\cd\,B)[\gamma])\cdot (a[\gamma])) \overset{\lam[]}{=}
  \El\,(\lam\,(\cd\,B[\gamma^+])\cdot (a[\gamma])) \overset{\Pi\beta}{=}
  \El\,(\cd\,B[\gamma^+][\langle a[\gamma]\rangle]) \overset{\cd[]}{=}
  \El\,(\cd\,B)[\gamma^+][\langle a[\gamma]\rangle] \overset{\U\beta}{=}
  B[\gamma^+][\langle a[\gamma]\rangle].
$
Hence, the conditional $\cdot[]'$ equation implies the condition for
all types.

Finally, we remove equation $[\p^+][\langle\q\rangle]$ by making
$\Pi\eta$ conditional. This needs another change: replacing $\Pi\beta$
with a more general variant conditional on the same equation.
\begin{alignat*}{10}
  & \Pi\eta' && : (e:B[\p^+][\langle\q\rangle] = B)\ra t =_e \lam\,(t[\p]\cdot\q) \hspace{3em}
  & \Pi\beta' && : (e:B[\p^+][\langle\q\rangle] = B)\ra (\lam\,b)[\p]\cdot\q =_e b
\end{alignat*}
But first we need to verify that the instance of $\Pi\beta$ used when
proving $B[\langle a\rangle][\gamma] = B[\gamma^+][\langle
  a[\gamma]\rangle]$ is derivable from $\Pi\beta'$. Assuming
$\hat{B} : \Tm\,(\Gamma\ext A)\,\U$, we argue
$
  \Pi\beta^\U :{} \lam\,\hat{B}\cdot a \overset{\q[\langle\rangle])}{=}
  \lam\,\hat{B}\cdot(\q[\langle a\rangle]) \overset{[\p][\langle\rangle]}=
  ((\lam\,\hat{B})[\p][\langle a\rangle])\cdot(\q[\langle a\rangle]) \overset{\cdot[]'\,\U[]}=
  ((\lam\,\hat{B})[\p]\cdot\q)[\langle a\rangle] \overset{\Pi\beta'\,\U[]}=
  \hat{B}[\langle a\rangle].
$
Now, for any type $B : \Ty\,(\Gamma\ext A)$, we derive the input of $\Pi\beta'/\Pi\eta'$ by
$
  B[\p^+][\langle\q\rangle] \overset{\U\beta}=
  \El\,(\cd\,B)[\p^+][\langle\q\rangle] \overset{\El[]}=
  \El\,(\cd\,B[\p^+][\langle\q\rangle]) \overset{\Pi\beta^\U}=
  \El\,(\lam\,(\cd\,B[\p^+])\cdot\q) \overset{\lam[]}= 
  \El\,(\lam\,(\cd\,B)[\p]\cdot\q) \overset{\Pi\beta'\,\U[]}=
  \El\,(\cd\,B) \overset{\U\beta}=
  B.
$
We summarise this section by formally stating the above.

\begin{definition}[Minimised single substitution calculus with $\Pi$ and $\U$]\label{def:min}
  This section defined the minimised version of Definition
  \ref{def:tt}, relying essentially on the presence of universes and
  $\Pi$ types. For reference, we list all the rules of the minimised
  calculus in Appendix \ref{app:minimisation}.
\end{definition}
\begin{problem}
  The GATs of Definition \ref{def:min} and Definition \ref{def:tt} are
  isomorphic, in particular, all equations are interderivable.
\end{problem}
\begin{proof}[Construction]
  Clear from the construction in this section.
\end{proof}

\section{\texorpdfstring{CwF from SSC with $\Sigma$, $\Pi$ and $\U$}{CwF from SSC with Σ, Π and U}}
\label{sec:cwf}

In this section we show that if the single substitution calculus has
certain type formers, then a CwF structure is derivable in any model.
The idea is that (i) using $\Sigma$ types we emulate contexts; (ii)
using functions between these $\Sigma$ types we emulate parallel
substitutions; (iii) using the functions into the universe we emulate
dependent types. Then single substitutions are just there to set up
$\Sigma$, $\Pi$ and $\U$, and these types are enough to bootstrap the
parallel substitution calculus.

\begin{problem}
  Every CwF (with type formers $\Pi, \dots, \Sigma$) gives rise to an SSC (with
  the same type formers).
\end{problem}
\begin{proof}[Construction]
  Most operations are the same, we set $\gamma^+ :=
  (\gamma\circ\p,\q)$ and $\langle a\rangle := (\id,a)$. All equations are derivable.
\end{proof}
The following construction is also known as the standard model,
metacircular interpretation \cite{DBLP:conf/popl/AltenkirchK16},
contextualisation \cite{DBLP:conf/fscd/BocquetKS23}. Following
\cite{DBLP:conf/mpc/KaposiKK19} we call it termification, as most sorts in the
new model are given by the sort of terms in the old model.
\begin{problem}[Termification]\label{prob:termification}
  From a model of SSC (with $\Pi, \dots, \Sigma$, see Definition
  \ref{def:sigma}), we define a CwF (with the same type formers).
\end{problem}
\begin{proof}[Construction]
  We define iterated lifting $\Lift^k :
  \Ty\,\Gamma\,i\ra\Ty\,\Gamma\,(k+i)$ by induction on $k$, together
  with $\un^k : \Tm\,(\Lift^k\,A)$ $\cong\Tm\,A : \mk^k$. The category % hack
  part of the CwF is given by types in the empty context and functions
  between them, suitably lifted (on the left hand side of $:=$ there
  is the component in the new CwF-model, on the right hand side the
  components refer to the old SSC-model):
  \[
    \Con := (i:\N)\times\Ty\,\diamond\,i \hspace{6em}
    \Sub\,\Delta\,\Gamma := \Tm\,\diamond\,(\Lift^{\Gamma-\Delta}\,\Delta\Ra\Lift^{\Delta-\Gamma}\,\Gamma)
  \]
  In the definition of $\Sub$, we did not write the projections for
  $\Con$, so $\Delta$ can mean $\Delta_{.1}$ or $\Delta_{.2}$, and we used the truncating subtraction of
  natural numbers. Composition of substitutions is quite involved, but
  it is just function composition written with explicit weakenings and
  appropriate (un)liftings. The category laws hold. The empty context is
  modelled by $\top$, its $\eta$ law holds via $\eta$ for $\top$.
  \begin{alignat*}{10}
%    & \gamma\circ\delta := \lam\Bigg(\mk^{\Theta-\Gamma}\,\bigg(\un^{\Delta-\Gamma}\,\Big(\gamma[\p]\cdot\mk^{\Gamma-\Delta}\,\big(\un^{\Theta-\Delta}\,(\delta[\p]\cdot\mk^{\Delta-\Theta}\,(\un^{\Gamma-\Theta}\,\q))\big)\Big)\bigg)\Bigg) \\
    & \gamma\circ\delta := \lam\,(\mk^{\Theta-\Gamma}\,(\un^{\Delta-\Gamma}\,(\gamma[\p]\cdot\mk^{\Gamma-\Delta}\,(\un^{\Theta-\Delta}\,(\delta[\p]\cdot\mk^{\Delta-\Theta}\,(\un^{\Gamma-\Theta}\,\q)))))) \\
    & \id := \lam\,\q \hspace{10em} \diamond := (0,\top) \hspace{10em} \epsilon := \lam\,(\mk^\Gamma\,\tt)
  \end{alignat*}
  Types are given by functions into $\U$, terms are dependent
  functions into the type, with lots of lifting adjustments.
  \begin{alignat*}{10}
    & \Ty\,\Gamma\,i && := \Tm\,\diamond\,\big(\Lift^{1+i-\Gamma}\,\Gamma\Ra\Lift^{\Gamma-(1+i)}\,(\U\,i)\big) \\
%    & \Tm\,\Gamma\,A && := \Tm\,\diamond\,\Bigg(\Pi\,(\Lift^{i-\Gamma}\,\Gamma)\,\bigg(\Lift^{\Gamma-i}\,\Big(\El\,\big(\un^{\Gamma-(1+i)}\,(A[\p]\cdot\mk^{1+i-\Gamma}\,(\un^{i-\Gamma}\,\q))\big)\Big)\bigg)\Bigg)
    & \Tm\,\Gamma\,A && := \Tm\,\diamond\,(\Pi\,(\Lift^{i-\Gamma}\,\Gamma)\,(\Lift^{\Gamma-i}\,(\El\,(\un^{\Gamma-(1+i)}\,(A[\p]\cdot\mk^{1+i-\Gamma}\,(\un^{i-\Gamma}\,\q))))))
  \end{alignat*}
  To make the notation readable, from now on, we will not write the
  lifting decorations or universe levels. We repeat the previous
  definitions again without writing decorations.
  \begin{alignat*}{10}
    & \Con && := \Ty\,\diamond                                                          && \epsilon && := \lam\,\tt                                                                      && \Gamma\ext A && := \Sigma\,\Gamma\,(\El\,(A[\p]\cdot\q)) \\
    & \Sub\,\Delta\,\Gamma && := \Tm\,\diamond\,(\Delta\Ra\Gamma)                       && \Ty\,\Gamma && := \Tm\,\diamond\,(\Gamma\Ra\U)                                                && (\gamma,a) && := \lam\,(\gamma[\p]\cdot\q,a[\p]\cdot\q) \\
    & \gamma\circ\delta && := \lam\,(\gamma[\p]\cdot(\delta[\p]\cdot\q)) \hspace{1.4em} && A[\gamma] && := \lam\,(A[\p]\cdot(\gamma[\p]\cdot\q))                                         && \p && := \lam\,(\fst\,\q) \\
    & \id && := \lam\,\q                                                                && \Tm\,\Gamma\,A && := \Tm\,\diamond\,(\Pi\,\Gamma\,(\El\,(A[\p]\cdot\q))) \hspace{1.4em}       && \q && := \lam\,(\snd\,\q) \\
    & \diamond && := \top                                                               && a[\gamma] && := \lam\,(a[\p]\cdot(\gamma[\p]\cdot\q))
  \end{alignat*}
  Context extension is given by $\Sigma$ types and pairing/projections
  by pairing/projections of $\Sigma$. All the CwF equations hold, for
  example we derive the functor law for type substitution in Appendix
  \ref{app:functor}. Type formers are added by adjusting them to handle contexts built up
  by $\Sigma$ types, for example,
  $
  \Pi\,A\,B := \lam\,(\cd\,(\Pi\,(\El\,(A[\p]\cdot\q))\,(\El\,(B[\p][\p]\cdot(\q[\p],\q)))))
  $
  and $\lam\,b := \lam\,(\lam\,(b[\p][\p]\cdot(\q[\p],\q)))$. All
  substitution laws and $\beta$/$\eta$-laws hold.
\end{proof}
In Appendix \ref{app:functor}, we show that the roundtrip CwF
$\longrightarrow$ SSC
$\xrightarrow{\text{termification}}$ CwF results in a
contextually isomorphic CwF.

\section{Conclusions and further work}
\label{sec:conclusion}

We described type theory in a minimalistic way, without referring to
categories or parallel substitutions. Our presentation has pedagogical
value, as illustrated in Section \ref{sec:tt}.
It is amazing that any type theory can be described with such a
minimal substitution calculus: the lifted equations
(\ref{eq:lifted1})--(\ref{eq:lifted4}) are not required. We use them
when proving properties of the syntax (such as normalisation), but
they are admissible.

In the future, we would like to investigate whether the calculus can
be minimised even more, e.g.\ to the degree that there is at most one
proof for any equation. This would be interesting for a possible
coherent syntax of type theory avoiding the need for truncation in the
setting of homotopy type theory.

\bibliographystyle{eptcs}
\bibliography{b}

\newpage

\appendix

\section{Detailed comparison of our SSC-calculus and B-systems}
\label{app:bsys}

In this section, we describe the relationship between B-systems
\cite{bc} and our Definition
\ref{def:tt} in detail.

B-systems are B-frames
together with substitution, weakening and generic element
operations. A B-frame is given by sets $B_i$, $\tilde{B}_i$ and
functions between them as shown below. \\
\begin{tikzpicture}
  \node (b0) at (4,0) {$\top$};
  \node (b1) at (5,0) {$B_1$};
  \node (b2) at (12,0) {$B_2$};
  \node (b3) at (15,0) {$\dots$};
  \node (c1) at (7,1) {$\tilde{B}_1$};
  \node (c2) at (14,1) {$\tilde{B}_2$};
  \draw[->,font=\scriptsize] (b1) edge node[below] {$\mathsf{ft}_0$} (b0);
  \draw[->,font=\scriptsize] (b2) edge node[below] {$\mathsf{ft}_1$} (b1);
  \draw[->,font=\scriptsize] (b3) edge node[below] {$\mathsf{ft}_2$} (b2);
  \draw[->,font=\scriptsize] (c1) edge node[above] {$\partial_1$} (b1);
  \draw[->,font=\scriptsize] (c2) edge node[above] {$\partial_2$} (b2);
\end{tikzpicture} \\
Using our notation, this is the following data ($\diamond$ is the empty context, $\ext$ is context extension). \\
\begin{tikzpicture}
  \node (b0) at (4,0) {$\top$};
  \node (b1) at (5,0) {$\Ty\,\diamond$};
  \node (b2) at (12,0) {$(A:\Ty\,\diamond)\times\Ty\,(\diamond\ext A)$};
  \node (b3) at (15,0) {$\dots$};
  \node (c1) at (7,1)  {$(A:\Ty\,\diamond)\times\Tm\,\diamond\,A$};
  \node (c2) at (14,1) {$\big((A:\Ty\,\diamond)\times(B:\Ty\,(\diamond\ext A))\big)\times\Tm\,(\diamond\ext A)\,B$};
  \draw[->,font=\scriptsize] (b1) edge node[below] {$\mathsf{fst}$} (b0);
  \draw[->,font=\scriptsize] (b2) edge node[below] {$\mathsf{fst}$} (b1);
  \draw[->,font=\scriptsize] (b3) edge node[below] {$\mathsf{fst}$} (b2);
  \draw[->,font=\scriptsize] (c1) edge node[above] {$\mathsf{fst}$} (b1);
  \draw[->,font=\scriptsize] (c2) edge node[above] {$\mathsf{fst}$} (b2);
\end{tikzpicture} \\
The substitution operation $\mathbb{S}$ for an $x:\tilde{B}_{n+1}$ is
a homomorphism of the slice B-frames $\mathbb{B}/\partial(x) \ra
\mathbb{B}/\mathsf{ft}(\partial(x))$. In our notation, $x = (\Gamma,
A, a)$ where $a : \Tm\,\Gamma\,A$, and $\mathbb{S}$ corresponds to the
following functions for $\Ty$ (a map between the bottom rows of the
diagrams):
\begin{alignat*}{10}
  & \blank[\langle a\rangle] && : \Ty\,(\Gamma\ext A)\ra\Ty\,\Gamma \\
  & \blank[\langle a\rangle^+] && : \Ty\,(\Gamma\ext A\ext B)\ra\Ty\,(\Gamma\ext B[\langle a\rangle]) \\
  & \blank[\langle a\rangle^{++}] && : \Ty\,(\Gamma\ext A\ext B\ext C)\ra\Ty\,(\Gamma\ext B[\langle a\rangle]\ext C[\langle a\rangle^+]) \\
  & \dots,
\end{alignat*}
and similarly $\mathbb{S}$ also includes all the (lifted) substitution
operations for terms (top rows of the
diagrams). Analogously, the weakening operation corresponds to
$\blank[\p^{+\dots+}]$ operations, and the generic element is $\q$ in
our notation. The six groups of equations correspond to our
$[\langle\rangle][]$, $[\p][^+]$, $\q[^+]$, $[\p][\langle\rangle]$,
$\q[\langle\rangle]$, $[\p^+][\langle\q\rangle]$ equations, in this
order. However we don't include the lifted versions of these equations,
equations $[\langle\rangle][]$, $[\p^+][\langle\q\rangle]$ are only
stated for types, and $\q[^+]$, $\q[\langle\rangle]$ are only stated
for terms. The missing equations are admissible, see Section
\ref{sec:admissible}. In the presence of $\U$ and $\Pi$, we
reduce the needed equations even more, see Section
\ref{sec:minimisation}.

\newpage

\section{Listings for Section \ref{sec:tt}}
\label{app:tt}

In Section \ref{sec:tt}, we presented the single substitution syntax
interleaved with explanations (Definition \ref{def:tt}), here we list
all the rules in one place, for reference.
\vspace{-0.5em}
\begin{multicols}{2}
\noindent The core substitution calculus:\vspace{-.5em}
\begin{alignat*}{10}
  & \Con && : \Set \\
  & \Ty && : \Con\ra\N\ra\Set \\
  & \diamond && : \Con \\
  & \blank\ext\blank && : (\Gamma:\Con)\ra\Ty\,\Gamma\,i\ra\Con \\
  & \Sub && : \Con\ra\Con\ra\Set \\
  & \Tm && : (\Gamma:\Con)\ra\Ty\,\Gamma\,i\ra\Set \\
  & \p && : \Sub\,(\Gamma\ext A)\,\Gamma \\
  & \langle\blank\rangle && : \Tm\,\Gamma\,A\ra\Sub\,\Gamma\,(\Gamma\ext A) \\
  & \blank^+ && : (\gamma:\Sub\,\Delta\,\Gamma)\ra\Sub\,(\Delta\ext A[\gamma])\,(\Gamma\ext A) \\
  & \blank[\blank] && : \Ty\,\Gamma\,i\ra\Sub\,\Delta\,\Gamma\ra\Ty\,\Delta\,i \\
  & \blank[\blank] && : \Tm\,\Gamma\,A\ra(\gamma:\Sub\,\Delta\,\Gamma)\ra\Tm\,\Delta\,(A[\gamma]) \\
  & \q && : \Tm\,(\Gamma\ext A)\,(A[\p]) \\
  & [\p][^+] && : B[\p][\gamma^+] = B[\gamma][\p] \\
  & [\p][^+] && : b[\p][\gamma^+] = b[\gamma][\p] \\
  & \q[^+] && : \q[\gamma^+] = \q \\
  & [\p][\langle\rangle] && : B[\p][\langle a\rangle] = B \\
  & [\p][\langle\rangle] && : b[\p ][\langle a\rangle] = b \\
  & \q[\langle\rangle] && : \q[\langle a\rangle] = a \\
  & [\langle\rangle][] && : B[\langle a\rangle][\gamma] = B[\gamma^+][\langle a[\gamma]\rangle] \\
  & [\p^+][\langle\q\rangle] && : B[\p^+][\langle\q\rangle] = B
\end{alignat*}

\newcolumn
\noindent Rules for individual type formers.\vspace{-.5em}
\begin{alignat*}{10}
  & \Pi && : (A:\Ty\,\Gamma\,i)\ra\Ty\,(\Gamma\ext A)\,i\ra\Ty\,\Gamma\,i \hspace{10em} \\
  & \Pi[] && : (\Pi\,A\,B)[\gamma] = \Pi\,(A[\gamma])\,(B[\gamma^+]) \\
  & \lam && : \Tm\,(\Gamma\ext A)\,B\ra\Tm\,\Gamma\,(\Pi\,A\,B) \\
  & \lam[] && : (\lam\,b)[\gamma] = \lam\,(b[\gamma^+]) \\
  & \blank\cdot\blank && : \Tm\,\Gamma\,(\Pi\,A\,B)\ra(a:\Tm\,\Gamma\,A)\ra \\
  & && \hphantom{{}:{}} \Tm\,\Gamma\,(B[\langle a\rangle]) \\
  & {\cdot}[] && : (t\cdot a)[\gamma] = (t[\gamma])\cdot(a[\gamma]) \\
  & \Pi\beta && : \lam\,b\cdot a = b[\langle a\rangle] \\
  & \Pi\eta && : t = \lam\,(t[\p]\cdot\q) \\
  & \U && : (i:\N)\ra\Ty\,\Gamma\,(1+i) \\
  & \U[] && : (\U\,i)[\gamma] = \U\,i \\  
  & \El && : \Tm\,\Gamma\,(\U\,i) \ra \Ty\,\Gamma\,i \\
  & \El[] && : (\El\,\hat{A})[\gamma] = \El\,(\hat{A}[\gamma]) \\
  & \cd && : \Ty\,\Gamma\,i\ra\Tm\,\Gamma\,(\U\,i) \\
  & \cd[] && : (\cd\,A)[\gamma] = \cd\,(A[\gamma]) \\
  & \U\beta && : \El\,(\cd\,A) = A \\
  & \U\eta && : \cd\,(\El\,\hat{A}) = \hat{A} \\
  & \Lift && : \Ty\,\Gamma\,i\ra\Ty\,\Gamma\,(1+i) \\
  & \Lift[] && : (\Lift\,A)[\gamma] = \Lift\,(A[\gamma]) \\
  & \mk && : \Tm\,\Gamma\,A \ra \Tm\,\Gamma\,(\Lift\,A) \\
  & \mk[] && : (\mk\,a)[\gamma] = \mk\,(a[\gamma]) \\
  & \un[] && : (\un\,a)[\gamma] = \un\,(a[\gamma]) \\
  & \un && : \Tm\,\Gamma\,(\Lift\,A) \ra \Tm\,\Gamma\,A \\
  & \Lift\beta && : \un\,(\mk\,a) = a \\
  & \Lift\eta && : \mk\,(\un\,a) = a
\end{alignat*}
\end{multicols}

\vspace{-0.5em}
\noindent We prove one of the round-trips in the isomorphism for lifting
variables (\ref{eq:liftvar}), the other is analogous:
\begin{alignat*}{10}
  & t[\p^+][\langle\un\,\q\rangle][\p^+][\langle\mk\,\q\rangle] && {=}(\ref{eq:lifted3})                                  &&    t[\p^+][{\p^+}^+][\langle\mk\,\q\rangle^+][\langle\un\,(\mk\,\q)\rangle] && {=}(\Lift\beta) \\
  & t[\p^+][{\p^+}^+][\langle\un\,\q[\p^+]\rangle][\langle\mk\,\q\rangle] && {=}(\un[])                                   &&    t[\p^+][{\p^+}^+][\langle\mk\,\q\rangle^+][\langle\q\rangle] && {=}(\ref{eq:lifted1}) \\      
  & t[\p^+][{\p^+}^+][\langle\un\,(\q[\p^+])\rangle][\langle\mk\,\q\rangle] && {=}(\q[^+])                                &&    t[\p^+][\p^+][\langle\mk\,\q\rangle^+][\langle\q\rangle] && {=}(\ref{eq:lifted2}) \\          
  & t[\p^+][{\p^+}^+][\langle\un\,\q\rangle][\langle\mk\,\q\rangle] && {=}(\ref{eq:lifted3})                              &&    t[\p^+][\langle\q\rangle] && {=}(\ref{eq:lifted4}) \\                                         
  & t[\p^+][{\p^+}^+][\langle\mk\,\q\rangle^+][\langle\un\,\q[\langle\mk\,\q\rangle]\rangle] && {=}(\un[])                &&    t                                                                                             \\
  & t[\p^+][{\p^+}^+][\langle\mk\,\q\rangle^+][\langle\un\,(\q[\langle\mk\,\q\rangle])\rangle] && {=}(\q[\langle\rangle]) \hspace{4em} && 
\end{alignat*}

\section{Listing for Section \ref{sec:minimisation}}
\label{app:minimisation}

For reference, we list the rules for the minimised single substitution
syntax (Definition \ref{def:min}). \\
\begin{minipage}{0.6\textwidth}
\begin{alignat*}{10}
  & \Con && : \Set \\
  & \Ty && : \Con\ra\N\ra\Set \\
  & \diamond && : \Con \\
  & \blank\ext\blank && : (\Gamma:\Con)\ra\Ty\,\Gamma\,i\ra\Con \\
  & \Sub && : \Con\ra\Con\ra\Set \\
  & \Tm && : (\Gamma:\Con)\ra\Ty\,\Gamma\,i\ra\Set \\
  & \p && : \Sub\,(\Gamma\ext A)\,\Gamma \\
  & \langle\blank\rangle && : \Tm\,\Gamma\,A\ra\Sub\,\Gamma\,(\Gamma\ext A) \\
  & \blank^+ && : (\gamma:\Sub\,\Delta\,\Gamma)\ra\Sub\,(\Delta\ext A[\gamma])\,(\Gamma\ext A) \\
  & \blank[\blank] && : \Ty\,\Gamma\,i\ra\Sub\,\Delta\,\Gamma\ra\Ty\,\Delta\,i \\
  & \blank[\blank] && : \Tm\,\Gamma\,A\ra(\gamma:\Sub\,\Delta\,\Gamma)\ra\Tm\,\Delta\,(A[\gamma]) \\
  & \q && : \Tm\,(\Gamma\ext A)\,(A[\p]) \\
  & [\p][^+]' && : (e:B[\p][\gamma^+] = B[\gamma][\p])\ra b[\p][\gamma^+] =_e b[\gamma][\p] \\
  & \q[^+]' && : (e:B[\p][\gamma^+] = B[\gamma][\p])\ra \q[\gamma^+] =_e \q \\
  & [\p][\langle\rangle]' && : (e:B[\p][\langle a\rangle] = B)\ra b[\p ][\langle a\rangle] =_e b \\
  & \q[\langle\rangle]' && : (e:B[\p][\langle a\rangle] = B)\ra \q[\langle a\rangle] =_e a \\
  & \Pi && : (A:\Ty\,\Gamma\,i)\ra\Ty\,(\Gamma\ext A)\,i\ra\Ty\,\Gamma\,i \\
  & \Pi[] && : (\Pi\,A\,B)[\gamma] = \Pi\,(A[\gamma])\,(B[\gamma^+]) \\
  & \lam && : \Tm\,(\Gamma\ext A)\,B\ra\Tm\,\Gamma\,(\Pi\,A\,B) \\
  & \lam[] && : (\lam\,b)[\gamma] = \lam\,(b[\gamma^+]) \\
  & \blank\cdot\blank && : \Tm\,\Gamma\,(\Pi\,A\,B)\ra(a:\Tm\,\Gamma\,A)\ra\Tm\,\Gamma\,(B[\langle a\rangle]) \\
  & {\cdot}[]' && : (e:B[\langle a\rangle][\gamma] = B[\gamma^+][\langle a[\gamma]\rangle])\ra \\
  & && \hphantom{{}:{}} (t\cdot a)[\gamma] =_e (t[\gamma])\cdot(a[\gamma]) \\
\end{alignat*}
\end{minipage}
\begin{minipage}{0.4\textwidth}
\begin{alignat*}{10}
  & \Pi\eta' && : (e:B[\p^+][\langle\q\rangle] = B)\ra \\
  & && \hphantom{{}:{}} t =_e \lam\,(t[\p]\cdot\q) \\
  & \Pi\beta' && : (e:B[\p^+][\langle\q\rangle] = B)\ra \\
  & && \hphantom{{}:{}} (\lam\,b)[\p]\cdot\q =_e b \\
  & \U && : (i:\N)\ra\Ty\,\Gamma\,(1+i) \\
  & \U[] && : (\U\,i)[\gamma] = \U\,i \\  
  & \El && : \Tm\,\Gamma\,(\U\,i) \ra \Ty\,\Gamma\,i \\
  & \El[] && : (\El\,\hat{A})[\gamma] = \El\,(\hat{A}[\gamma]) \\
  & \cd && : \Ty\,\Gamma\,i\ra\Tm\,\Gamma\,(\U\,i) \\
  & \cd[] && : (\cd\,A)[\gamma] = \cd\,(A[\gamma]) \\
  & \U\beta && : \El\,(\cd\,A) = A \\
  & \U\eta && : \cd\,(\El\,\hat{A}) = \hat{A} \\
  & \Lift && : \Ty\,\Gamma\,i\ra\Ty\,\Gamma\,(1+i) \\
  & \Lift[] && : (\Lift\,A)[\gamma] = \Lift\,(A[\gamma]) \\
  & \mk && : \Tm\,\Gamma\,A \ra \Tm\,\Gamma\,(\Lift\,A) \\
  & \mk[] && : (\mk\,a)[\gamma] = \mk\,(a[\gamma]) \\
  & \un[] && : (\un\,a)[\gamma] = \un\,(a[\gamma]) \\
  & \un && : \Tm\,\Gamma\,(\Lift\,A) \ra \Tm\,\Gamma\,A \\
  & \Lift\beta && : \un\,(\mk\,a) = a \\
  & \Lift\eta && : \mk\,(\un\,a) = a \\
  & \\
  & \\
  & \\
\end{alignat*}
\end{minipage}

\newpage

\section{Listings for Section \ref{sec:cwf}}
\label{app:functor}

As an example of the equations in the termification model construction
(Problem \ref{prob:termification}), we derive one of the functor laws
for types:
\begin{equation*}
\begin{alignedat}{10}
  & A[\gamma\circ\delta] && {=} \\
  & \lam\,(A[\p]\cdot(\lam\,(\gamma[\p]\cdot(\delta[\p]\cdot\q))[\p]\cdot\q)) && {=}(\lam[]) \\
  & \lam\,(A[\p]\cdot(\lam\,(\gamma[\p][\p^+]\cdot(\delta[\p][\p^+]\cdot\q))\cdot\q)) && {=}(\Pi\beta, {\cdot}[]) \\
  & \lam\,(A[\p]\cdot(\gamma[\p][\p^+][\langle\q\rangle]\cdot(\delta[\p][\p^+][\langle\q\rangle]\cdot\q))) && {=}([\p][^+]) \\
  & \lam\,(A[\p]\cdot(\gamma[\p][\p][\langle\q\rangle]\cdot(\delta[\p][\p][\langle\q\rangle]\cdot\q))) && {=}([\p][\langle\rangle]) \\
  & \lam\,(A[\p]\cdot(\gamma[\p]\cdot(\delta[\p]\cdot\q))) && {=}(\Pi\beta) \\
  & \lam\,(\lam\,(A[\p][\p^+]\cdot(\gamma[\p][\p^+]\cdot\q))\cdot(\delta[\p]\cdot\q))\,\, && {=}(\lam[]) \\
  & \lam\,(\lam\,(A[\p]\cdot(\gamma[\p]\cdot\q))[\p]\cdot(\delta[\p]\cdot\q)) && {=} \\
  & A[\gamma][\delta]
\end{alignedat}
\end{equation*}

We show that the roundtrip CwF $\longrightarrow$ SSC
$\xrightarrow{\text{termification}}$ CwF results in a contextually
isomorphic CwF. A contextual isomorphism
\cite{DBLP:journals/corr/abs-2211-07487} is a weak CwF-morphism
(pseudomorphism, i.e.\ context extension and the empty context are
only preserved up to isomorphism) which is bijective on types and
terms. Note that a contextual isomorphism preserves all type formers
specified by universal properties.
\begin{problem}
  Given a CwF (with type formers) $M$, let $M'$ denote the CwF
  obtained by first seeing it as an SSC model and then termifying it. We
  construct a contextual isomorphism between $M'$ and $M$.
\end{problem}
\begin{proof}[Construction]
  We denote the components of the contextual isomorphism $F$ as follows.
  \begin{alignat*}{10}
    & F : \Con_{M'}\ra\Con_M && F : \Ty_{M'}\,\Gamma \cong \Ty_M\,(F\,\Gamma) \\
    & F : \Sub_{M'}\,\Delta\,\Gamma\ra\Sub_M\,(F\,\Delta)\,(F\,\Gamma)\hspace{3em} && F : \Tm_{M'}\,\Gamma\,A \cong \Tm_M\,(F\,\Gamma)\,(F\,A)
  \end{alignat*}
  We define them in the same order, omitting the $M$ subscripts:
  $F\,\Gamma := \diamond\ext\Gamma$, 
  $F\,\gamma := (\p,\gamma[\p]\cdot\q)$, 
  $F\,A := \El\,(A[\p]\cdot\q)$, 
  $F\,a := a[\p]\cdot\q$.
%  \begin{alignat*}{10}
%    & F : \Ty\,\diamond\ra\Con && F : \Tm\,\diamond\,(\Gamma\Ra\U) \cong \Tm\,(\diamond\ext\Gamma)\,\U \cong \Ty\,(\diamond\ext\Gamma) \\
%    & F\,\Gamma := \diamond\ext\Gamma && F\,A := \El\,(A[\p]\cdot\q) \\
%    & F : \Tm\,\diamond\,(\Delta\Ra\Gamma) \ra \Sub\,(\diamond\ext\Delta) (\diamond\ext\Gamma) && F : \Tm\,\diamond\,(\Pi\,\Gamma\,(\El\,(A[\p]\cdot\q))) \cong \Tm\,(\diamond\ext\Gamma)\,(\El\,(A[\p]\cdot\q)) \\
%    & F\,\gamma := (\p,\gamma[\p]\cdot\q) \hspace{3em} && F\,a := a[\p]\cdot\q
%  \end{alignat*}
  It is easy to check that $F$ preserves the CwF structure, i.e.\ it
  is a functor, $\epsilon : \Sub\,(F\,\diamond)\,\diamond$ is an
  isomorphism, $F\,(A[\gamma]) = F\,A[F\,\gamma]$, $F\,(a[\gamma]) =
  F\,a[F\,\gamma]$ and $(F\,\p,F\,\q) : \Sub\,(F\,(\Gamma\ext
  A))$ $(F\,\Gamma\ext F\,A)$ is an isomorphism. % hack
\end{proof}
The other round-trip, that is, starting with an SSC-model, termifying
it and comparing the result with the original SSC-model does not provide an
isomorphism because we can't define the map $F :
\Tm\,\diamond\,(\Delta\Ra\Gamma) \ra
\Sub\,(\diamond\ext\Delta)\,(\diamond\ext\Gamma)$ on
substitutions. Such a map is not definable only using single
substitutions (a $\Sub$ is either a lifted weakening or a lifted single substitution, but we don't know anything about the relationship between $\Delta$ and $\Gamma$). However, the set of terms and types are still
isomorphic after the round-trip. We leave formulating the right notion
of contextual isomorphism for SSCs as future work.

\end{document}